%% file: TopoSortCausalDisco.tex
\documentclass{article}

\usepackage[main,final,nonatbib]{neurips_2026}
\input{chapters/preamble.tex}

\title{A Topological Sorting Criterion for Random Causal Directed Acyclic Graphs}

\author{%
    Alexander G.\ Reisach\\
    Laboratoire de Mathématiques d'Orsay\\
    Université Paris-Saclay\\
    91400 Orsay, France
    \And
    Antoine Chambaz \\
    $\phantom{blablablabbla}$CNRS, MAP5$\phantom{blablablablba}$ \\
    Université Paris Cité\\
    F-75006 Paris, France
    \AND
    Gilles Blanchard\thanks{Equal senior author contribution.}\\
    Laboratoire de Mathématiques d'Orsay\\
    Université Paris-Saclay\\
    91400 Orsay, France
    \And
    Sebastian Weichwald$^*$\\
    Department of Mathematical Sciences and\\
    Pioneer Centre for Artificial Intelligence,\\
    University of Copenhagen, Denmark
}

\begin{document}

\maketitle

\begin{refsection}
\input{chapters/abstract.tex}
\input{chapters/introduction.tex}
\input{chapters/theory.tex}
\input{chapters/results.tex}
\input{chapters/MEC_size.tex}
\input{chapters/time_DAGs.tex}
\input{chapters/discussion.tex}
\input{chapters/conclusion.tex}
\ack
Alexander G.\ Reisach and Gilles Blanchard were supported by the ANR grant BISCOTTE (ANR-19-CHIA-0021). Sebastian Weichwald was supported by a research grant (0069071) from Novo Nordisk Fonden, and this research was supported by the Pioneer Centre for AI, DNRF grant number P1.
\printbibliography
\end{refsection}

\begin{refsection}
\appendix
\crefalias{section}{appendix}
\input{chapters/appendix.tex}
\printbibliography[title={Appendix References}]
\end{refsection}
\newpage

\end{document}

%% file: chapters/preamble.tex
\usepackage{amsmath}
\usepackage{amssymb}
\usepackage{amsthm}
\usepackage{mathtools}
\usepackage{graphicx}
\usepackage{xspace}
\usepackage{centernot}
\usepackage{graphicx}
\usepackage{float}
\usepackage{subcaption}
\usepackage[ruled,vlined]{algorithm2e}
\usepackage{xcolor}
\usepackage{microtype}
\usepackage{booktabs}       %
\usepackage{amsfonts}       %
\usepackage{nicefrac}       %
\usepackage{blkarray}
\definecolor{linkcolor}{HTML}{449B52}
\definecolor{citecolor}{HTML}{67001F}
\definecolor{urlcolor}{HTML}{008080}
\usepackage[breaklinks,colorlinks=true,
            citecolor=citecolor,urlcolor=urlcolor,linkcolor=linkcolor]{hyperref}
\usepackage[capitalize,noabbrev]{cleveref}
\usepackage{tikz}
\usetikzlibrary{decorations.pathreplacing}
\usepackage[
    backend=biber,
    uniquename=false,
    bibencoding=utf8,
    sorting=nty,
    doi=false,isbn=false,url=false,
    citestyle=authoryear,
    maxcitenames=1,
    uniquelist=false,
    maxbibnames=10,
    minbibnames=10,
    backref=true,
    hyperref=true
]{biblatex}
\AtEveryBibitem{\clearname{editor}}
\DeclareFieldFormat[article,inproceedings,manual,book,misc,thesis,incollection]{title}{%
  \iffieldundef{url}{#1}{\href{\thefield{url}}{#1}}}

\let\cite\textcite

\renewcommand{\Pr}{\mathbb{P}} 
\renewcommand{\P}{\mathcal{P}}

\newcommand{\G}{\mathcal{G}}
\newcommand{\Exp}[1]{\mathbb{E}\left[#1\right]}

\newcommand{\V}{\mathcal{V}}
\newcommand{\E}{\mathcal{E}}

\newcommand{\D}{\mathcal{D}}
\newcommand{\C}{\mathcal{C}}
\newcommand{\U}{\mathcal{U}}
\newcommand{\R}{\mathbb{R}}
\newcommand{\N}{\mathbb{N}}
\newcommand{\X}{\mathbf{X}}
\newcommand{\s}{\mathbf{s}}
\newcommand{\argmin}{\mathop{\arg\min}}
\newcommand{\ERname}{Erd\H{o}s--R\'enyi\xspace}

\renewcommand{\v}{\boldsymbol{\mathrm{v}}}
\newcommand{\roots}{\operatorname{roots}}
\newcommand{\anc}{\operatorname{anc}}
\newcommand{\desc}{\operatorname{desc}}
\newcommand{\pa}{\operatorname{pa}}
\newcommand{\ch}{\operatorname{ch}}
\newcommand{\rel}{\operatorname{rel}}

\newcommand{\indep}{\perp\!\!\!\perp}
\newcommand{\nindep}{\centernot{\indep}}

\newtheorem{lemma}{Lemma}
\newtheorem{proposition}{Proposition}
\newtheorem{theorem}{Theorem}
\newtheorem{corollary}{Corollary}
\theoremstyle{definition}
\newtheorem{definition}{Definition}
\theoremstyle{plain}

%% file: chapters/abstract.tex
\begin{abstract}
    \noindent
    Random directed acyclic graphs (DAGs) based on imposing a vertex order on \ERname or scale-free graphs are widely used to generate synthetic data from causal models, in particular for evaluating causal discovery algorithms.
    We show that the set of nodes reachable via open paths, termed \textit{relatives}, never decreases along the causal order in all DAGs.
    We find that it tends to increase strictly for DAGs with certain density parameters, and demonstrate that this pattern can be exploited for causal discovery via sorting by the estimated number of relatives. 
    We note that many simulations in the literature are focused on a narrow range of parameter settings where sorting by relatives yields an excellent proxy for the causal order, and show that a strict increase of relatives leads to a singleton Markov equivalence class.
    We propose sampling time-series DAGs as a possible alternative and discuss implications for causal discovery and synthetic data generation. Implementations of our algorithms can be found in the \href{https://causaldisco.github.io/CausalDisco/}{CausalDisco} Python package.
\end{abstract}

%% file: chapters/introduction.tex
\section{Introduction}

\paragraph{Random DAGs in causal discovery.}
In the field of statistical causality, directed acyclic graphs (DAGs) are a mathematical tool used to encode the causal relationships between a set of random variables \parencite{pearl2009causality}. 
Causal discovery aims to learn the underlying causal DAG or an equivalence class from observations of these random variables \parencite{spirtes2001causation}.
To evaluate the performance of causal discovery algorithms, most works in the literature evaluate them on synthetically generated data. We list examples in \cref{tab:literature} and refer to \cite{heinze2018causal,vowels2022DAGs,kitson2023survey} for extensive overviews. We further note that similar causal synthetic data generation mechanisms are also used beyond causal discovery, for example in tabular foundation models \parencite{hollmann2025accurate,qu2026tabiclv2}. 
Synthetic data generation schemes in causal discovery rely overwhelmingly on the \ERname \parencite{erdos1960evolution} and scale-free \parencite{barabasi1999emergence} random graph models, whose edges are then oriented along a vertex order to obtain DAGs: a uniformly random order for \ERname graphs, and the reverse of the order in which nodes were added for scale-free graphs.
For convenience, we will refer to DAGs derived from \ERname graphs in this way as \textit{ER DAGs}, and to DAGs derived from scale-free graphs in this way as \textit{SF DAGs}.
Despite their popularity, the topological properties of these random DAGs have as of yet received little study within the context of causal discovery.

\paragraph{Related literature.}
Our contribution is motivated by a line of recent work questioning the realism of synthetic benchmarks in causal discovery \parencite[cf.][]{brouillard2025landscape}.
\cite{gentzel2019case} highlight the arbitrary nature of synthetic data generation and advocate evaluation on downstream causal reasoning tasks.
\cite{reisach2021beware} show that variance tends to explode along the causal order in synthetic data and introduce \textit{var-sortability} to capture the effect. \cite{reisach2023scale} show that by the same mechanism, causal relationships effectively become deterministic in large graphs, as captured by \textit{$R^2$-sortability}. 
Following these findings, \cite{andrews2024better}, \cite{ormaniec2024standardizing}, and \cite{herman2025unitless} suggest alternative data generation regimes with better-understood sortability properties.
Nonetheless, the graph sampling regimes have so far gone largely unchallenged, and ER DAGs and SF DAGs remain the de facto standard in the literature.
Our contribution extends the scrutiny of synthetic data generation to the topology of the underlying graph models and provides theoretical results to guide the evaluation of causal discovery algorithms using synthetic data.
Thus, our work is also closely related to \cite{nagarajan2025benchmarking}, which investigates the effect of different preferential attachment rules on constraint-based causal discovery, and to \cite{blanchard2025phase}, which studies the connectivity properties of \ERname graphs with a vertex ordering.
The topology of a random DAG also determines the size of its Markov equivalence class (MEC) \parencite{schmid2022number}. Our results link this structural property to sortability by relatives.
Thus, our work aims to bridge the theoretical analysis of random graphs and empirical evaluation of causal discovery algorithms by introducing a new sortability concept that captures topological properties of random DAGs that are relevant to causal discovery.

\paragraph{Contribution.}
In \cref{sec:stb_by_rels}, we introduce the notion of relatives and provide theoretical results indicating that it can capture aspects of random DAGs that are relevant for causal discovery. We then analyze the agreement between an ordering by the number of relatives and a causal ordering numerically through \textit{rel-sortability}, and find that it is high for common settings in the causal discovery literature.
In \cref{sec:discovery}, we illustrate that the number of relatives can be exploited for causal discovery and is largely unaffected by data generation schemes designed to remove other data artifacts.
We also observe that a state-of-the-art MEC discovery algorithm yields excellent performance on DAG discovery, likely due to small MECs.
Motivated by this finding, in \cref{sec:MEC} we present a theoretical result linking a strictly increasing number of relatives along the causal order to singleton MECs.
Since the pattern of increasing number of relatives along the causal order is likely unintentional and its prevalence may vary with application domains, in \cref{sec:time_series} we turn to time-series DAGs as a potential alternative for generating random DAGs and provide a sufficient condition to control sortability by relatives for them.
In \cref{sec:discussion} we discuss limitations, opportunities for causal discovery, and implications for synthetic data generation using causal DAGs. We state our conclusion in \cref{sec:conclusion}.

%% file: chapters/theory.tex
\section{Sortability by relatives}\label{sec:stb_by_rels}
In \cref{sec:definitions} we give a formal definition of causal DAGs, introduce the concept of relatives, and state a monotonicity property.
In \cref{sec:sortability} we first define a measure of agreement between a causal ordering and the ordering by the number of relatives, and then prove a corresponding lower bound for ER DAGs.
We assess the agreement between causal order and ordering by relatives numerically for different graph sizes and sparsities in \cref{sec:numerical_prevalence}.
In \cref{sec:real_world_relstb}, we evaluate this agreement on real-world DAGs.

\subsection{Definitions}\label{sec:definitions}

\paragraph{Causal DAGs.}
Let $\G=(\V, \E)$ be a graph of node set $\V$ and edge set $\E\subseteq\V\times\V$. Unless otherwise indicated, we denote vertices by lower-case letters.
We consider the case of directed graphs, meaning that $(x,y)\in \E$ denotes the directed edge $x\to y$.
An ordered collection of subsequent edges forms a directed path. 
We assume acyclicity of the graph, meaning that there are no directed paths 
from a node to itself.
$\G$ is therefore a DAG.
Let each node in $\V$ correspond to a random variable with joint law $\Pr$.
We say the law $\Pr$ is compatible with $\G$ if it is Markov with respect to $\G$, and we assume that such a law is also faithful with respect to $\G$
\parencite[see e.g.][Definitions 6.21 and 6.33]{peters2017elements}.
Thus, the $d$-separation criterion \parencite[][Section 4]{geiger1990identifying} in $\G$ matches the conditional independence relations in $\Pr$, and we refer to the graph as a probabilistic graphical model of the random variables.
Assume now that a directed edge $(x,y)$ corresponds to $x$ being a cause of $y$ in some external system modelled by the graph \parencite[][Section 5]{dawid2010beware}. Then, $\G$ is a graphical causal model of the random variables \parencite[][Definition 2.2.2]{pearl2009causality}. For our causal discovery experiments, we additionally assume causal sufficiency, meaning there are no unobserved confounders.

\paragraph{The notion of relatives.}
For every $a,z\in\V$, let
\begin{equation}
    \P_{az} = \begin{cases}
        \text{set of directed paths from }a\text{ to }z,\; &\text{if }a\neq z\\
        \{a\}\text{ by convention},\; &\text{if }a=z
    \end{cases}
\end{equation}
be the set of paths from $a$ to $z$.
On this basis we define the following shorthands for any $x,y,z\in\V$.
Let $\anc(y)\coloneq\{x\colon\P_{xy}\neq\emptyset\}$ denote the ancestors of a given node $y$, let $\desc(y)\coloneq \{z\colon\P_{yz}\neq\emptyset\}$ denote its descendants, and let $\roots(y)\coloneq\{x\colon x\in\anc(y), \anc(x)=\{x\}\}$ denote its roots. We further define the parents $\pa(y)=\{x\colon (x,y)\in\E\}$ and children $\ch(x)=\{y\colon (x,y)\in\E\}$.
Whenever the argument to $\anc$, $\desc$, $\rel$, $\pa$, or $\ch$ is itself a set, we take this to mean the union of the respective operator on each element of the set. This allows us to define relatives as follows.
\begin{definition}[Relatives]
    \label{def:relatives}
    For a node $y$ in a causal DAG, the relatives of $y$ are
    \begin{align*}
        \rel(y)\coloneq\underset{x\in\anc(y)}{\bigcup}\desc(x) = \desc(\anc(y)).
    \end{align*}
\end{definition}
\noindent
A brief illustration of the notion of relatives can be found in \cref{app:relatives}.
Since we assume faithfulness, $x\in\rel(y)$ implies $x\nindep y \mid \emptyset$, meaning that the relatives of a node other than itself are exactly the adjacencies found in the first iteration of the PC algorithm \parencite[][Section 5.4.2]{spirtes2001causation}.

\paragraph{Relevance for causal discovery.}
The relatives of a node are interesting for the purpose of causal discovery primarily due to the monotonicity property in \cref{thm:relatives}, which suggests that an ordering based on the estimated number of relatives in increasing fashion may result in a good approximation of a causal ordering.
\begin{theorem}[Monotonicity of relatives]
    \label{thm:relatives}
    For any $x,y\in\V$,%
    \begin{align}
        x\in \anc(y)\implies \rel(x) \subseteq \rel(y).
    \end{align}
\end{theorem}
\begin{proof}
    Note that
        $\rel(y) = \desc(\anc(y)) = \desc\left(\anc(x)\right) \cup \desc\left(\anc(y)\setminus\anc(x)\right),$
    where the first argument of the union on the right-hand side is equal to $\rel(x)$, hence the implication.
\end{proof}

\subsection{Sorting by the number of relatives}\label{sec:sortability}

To capture the agreement between a causal order and an ordering by the number of relatives, we propose a \textit{sortability} criterion in the sense of \textcite[][Section 3.1]{reisach2023scale}. Note that causal orders are partial orders so classical measures of order agreement are not easily applicable, motivating sortability.
Let $\G=(\V,\E)$ with $|\E|\geq1$ be a DAG, and let $\rho\colon \V\to\R$ be a function to be used as a node sorting criterion.
We define sortability of $\G$ by $\rho$ as
\begin{align}
    \s_\rho(\G) \coloneq \frac{\sum_{(x,y)\in\E}\s_\rho\big((x,y)\big)}{|\E|}
    \quad \text{with} \quad
    \s_\rho\big((x,y)\big) \coloneq \begin{cases}
        1, &\text{if } \rho(x)<\rho(y)\\
        \frac{1}{2}, &\text{if } \rho(x)=\rho(y)\\
        0, &\text{if } \rho(x)>\rho(y)
    \end{cases}
    \quad \text{for } (x,y)\in\E.
    \label{eq:sortability}
\end{align}
We take $\rho$ to denote the number of relatives of a given node $v\in\V$ in $\G$, that is, $\rho(v)\coloneq|\rel(v)|$. As a shorthand, we refer to $\s_\rho(\G)$ as the \textit{rel-sortability} of $\G$.
For $\G$ the ground-truth DAG of a causal model, \cref{eq:sortability} hence measures alignment between the causal order and the weak order resulting from sorting by the number of relatives. Later on, we will show how the number of relatives can be inferred from data, and assess sortability of the true DAG by  $\hat\rho$, the empirically estimated number of relatives. 
As with the sortabilities in \cite{reisach2021beware} and \cite{reisach2023scale}, for a $\s_\rho$ value of $1$ sorting by $\rho$ reveals a causal ordering, for a value of $0$ it reveals an inverse causal ordering, and for a value of $1/2$ it amounts to a random guess.
Note that our definition differs from the sortabilities in these works in that we count edges instead of paths.
We prefer this version since it retains the qualitative interpretation (cf.\ our forthcoming experiments and \cref{app:path_length_stb}) and can also be understood as an inverted and normalized version of the topological order divergence in \textcite[][Section 4.1]{rolland2022score}, provided there are no ties in $\rho$.
With ties in $\rho$ broken uniformly at random, $(1-\s_\rho(\G))|\E|$ is the expected number of edges of $\G$ that contradict the resulting order, and thus a lower bound on the expected Hamming distance between $\G$ and any graph compatible with that order.

\paragraph{A lower bound on rel-sortability in ER DAGs.}
\cref{thm:relatives} tells us that rel-sortability cannot fall below $1/2$. A tighter lower bound for edges in ER DAGs is given in \cref{thm:lower_bound}. 
\begin{theorem}
    \label{thm:lower_bound}
    Consider an ER DAG obtained from an $ER(n,p)$ graph with $p=2c/(n-1)$ for a density parameter $c>0$ by imposing a fixed topological order.
    For any node $v\in\V$, let $r_v\in\{1,\dots,n\}$ be the rank of $v$ in the order, with $r_v=1$ for the first node. For any fixed $0<q_x<q_y<1$, let $x$ and $y$ be the nodes with ranks $r_x=\lceil q_x n\rceil$ and $r_y=\lceil q_y n\rceil$, and assume $(x,y)\in\E$. Then
    \begin{align}
        \mathop{\lim\inf}_{n\to\infty}\Exp{\s_\rho\big((x,y)\big)}
        \geq
        \frac12+\frac12\left(1-\exp\left(e^{-2c q_y}-e^{-2c q_x}\right)\right).
        \label{eq:lower_bound}
    \end{align}
\end{theorem}
\noindent The proof of \cref{thm:lower_bound} rests on \cref{lem:ancestors} in \cref{app:proof_details} and is shown in \cref{app:thm_2}.
Note that the lower bound in \cref{eq:lower_bound} increases with the quantile distance between $x$ and $y$, that is, when $q_y$ is increased or $q_x$ is decreased. For fixed $q_x<q_y$, its dependence on the density parameter $c$ is non-monotone: it first increases and then decreases as $c$ grows, approaching $1/2$ in both limits. Hence, we may expect rel-sortability in ER DAGs to exceed the generic lower bound of $1/2$ in intermediate-density regimes, while approaching $1/2$ in both very sparse and very dense regimes.

\subsection{Numerical analysis of the prevalence of rel-sortability}\label{sec:numerical_prevalence}
We begin by probing the prevalence of high rel-sortability across random DAGs numerically by simulating graphs with different numbers of nodes $n$ and density parameter $c$. For ER DAGs, we use the parametrization of \cref{thm:lower_bound}, and for SF DAGs, each new node attaches to $c$ existing nodes, so that both have approximately $cn$ edges. \cref{fig:sortability_er_sf} shows the results for ER and SF DAGs with $n\in(2,\dots,100)$ and $c\in(0.5,1,2, \dots, 50)$, for all combinations with $p\leq1$. Each sortability value shown is an average over $10$ independent draws for each combination of $n$ and $c$.

\paragraph{Results.}
The result for ER DAGs can be seen in \cref{fig:sortability_er}. We observe a band of high sortability in the intermediate density regime from $c=1$ until about $c=5$, with the upper end growing slowly in $n$.
In \cref{fig:sortability_sf}, we observe a similar pattern for SF graphs with a wider band of high rel-sortability.
We stress that many simulation settings used in the causal discovery literature fall into the high rel-sortability bands (see \cref{app:graphs_lit}).
The results for ER DAGs match our qualitative analysis of the rel-sortability lower bound in \cref{thm:lower_bound}, and the similarity between \cref{fig:sortability_er} and \cref{fig:sortability_sf} suggests a similar bound may exist for SF DAGs. Overall, the experiment indicates that estimating the number of relatives has the potential to yield a close approximation of a causal ordering.
\vskip-.85em
\begin{figure}[H]
    \centering
    \begin{subfigure}{.49\linewidth}
        \centering
        \includegraphics[width=\linewidth]{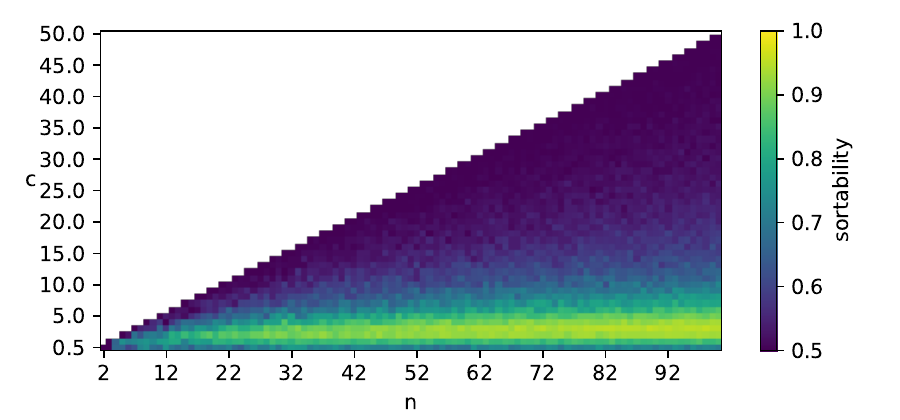}
        \vskip-.5em
        \caption{ER DAGs}
        \label{fig:sortability_er}
    \end{subfigure}
    \begin{subfigure}{.49\linewidth}
        \centering
        \includegraphics[width=\linewidth]{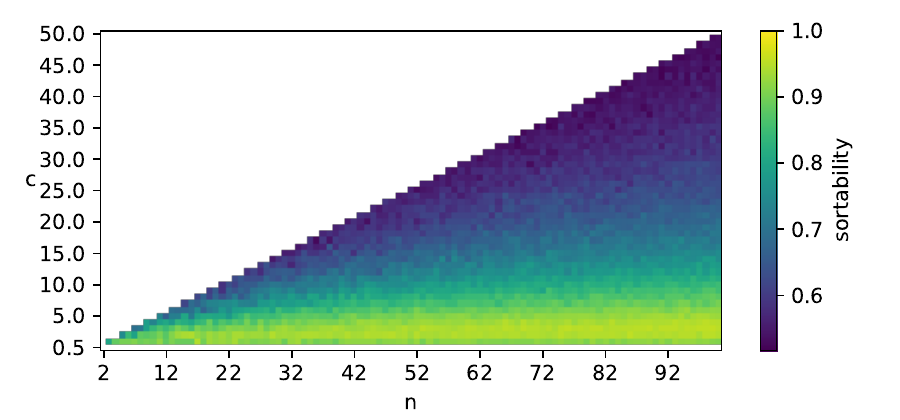}
        \vskip-.5em
        \caption{SF DAGs}
        \label{fig:sortability_sf}
    \end{subfigure}
    \vskip-.2em
    \caption{Rel-sortability of random ER and SF DAGs with $n$ nodes and density parameter $c$.}
    \vskip-.05em
    \label{fig:sortability_er_sf}
\end{figure}

\subsection{Rel-sortability on real-world datasets}\label{sec:real_world_relstb}
Our numerical analysis of rel-sortability in synthetic data naturally raises the question of rel-sortability in real-world data. To provide a starting point, we have evaluated the rel-sortability of datasets from the bnlearn repository \parencite[\textit{BNL},][]{scutari2010learning} and the causalAssembly datasets \parencite[\textit{CA},][]{gobler2024texttt}. For CA, we consider each of its five production stations (S1 to S5) and the full production line. As shown in \cref{tab:real_world_rel_stb}, rel-sortabilities vary widely, from $0.5$ to nearly $1$, motivating our forthcoming investigation into the effects of rel-sortability and how to control it in synthetic data.
\vskip-.6em
\begin{table}[H]
    \centering
    \footnotesize
    \setlength{\tabcolsep}{3pt}
    \input{figures/real_world_rel_sortability.tex}
    \caption{Rel-sortability of selected real-world DAGs.}
    \label{tab:real_world_rel_stb}
\end{table}

%% file: figures/real_world_rel_sortability.tex
\begin{tabular}{l|lrrrrrr|lrrrrr}
\toprule
 & \textbf{CA} & S1 & S3 & S5 & S4 & S2 & Full & \textbf{BNL} & sprinkler & asia & sachs & alarm & andes \\
\midrule
nodes, edges &  & 6, 7 & 16, 11 & 16, 16 & 26, 47 & 34, 77 & 98, 485 &  & 4, 4 & 8, 8 & 11, 17 & 37, 46 & 223, 338 \\
rel-sortability &  & 0.71 & 0.68 & 0.81 & 0.73 & 0.90 & 0.96 &  & 0.50 & 0.69 & 0.50 & 0.86 & 0.96 \\
\bottomrule
\end{tabular}

%% file: chapters/results.tex
\section{Exploiting rel-sortability for causal discovery}\label{sec:discovery}
We intend to compare sortability by true number of relatives to sortability by an empirical estimate of that number, and to use the latter for causal discovery. In \cref{sec:algorithms} we propose an algorithm to estimate relatives from data, and a second one that uses the induced order to recover a DAG. 
Then, in \cref{sec:schemes} we compare rel-sortability by true and estimated relatives to other sortability criteria in the literature.
Finally, in \cref{sec:emp_causal_discovery}
we show that sorting by estimated relatives can be highly effective for causal discovery for common values of the density parameter $c$, and also note that a state-of-the-art MEC discovery algorithm run for comparison discovers DAGs nearly perfectly.

\subsection{Algorithms}\label{sec:algorithms}
\paragraph{Estimating an ordering by the number of relatives empirically.}
Let $\G$ be a causal DAG of $n$ nodes, and let the random vector $(X_1,\dots,X_n)$ contain the random variables corresponding to the nodes in $\G$.
Let $\X_j\in\R^m$ be a real-valued vector of $m$ independently and identically distributed (i.i.d.) observations of the $j$-th random variable in this vector. Given a threshold $\varepsilon>0$, we estimate the number of relatives empirically and use it to obtain an order as follows.\\
\DontPrintSemicolon
\SetNlSty{textbf}{\small}{.}
\begin{algorithm}[H]
    \setlength{\baselineskip}{1.2\baselineskip}  %
    \KwIn{data $(\X_1, \dots, \X_n)$, 
    threshold $\varepsilon$.
          }
    Compute the correlation matrix $C\in\R^{n\times n}$ with $C_{ij}=\frac{\text{Cov}(\X_i,\X_j)}{\sqrt{\text{Var}(\X_i)\text{Var}(\X_j)}}$ for $i,j\in (1,\dots,n)$.\\
    Let $\hat\rho\in\N^n$ be the number of column entries in $C$ with absolute value greater than $\varepsilon$.\\
    Define an ordering $s$ to be the index set given by sorting $(1,\dots,n)$ by $\hat\rho$ in ascending fashion.\\
    \textbf{Return} $s$.\\
    \caption{Estimating an ordering by the number of relatives.}
    \label{algo:est_rels}
\end{algorithm}
\noindent
Note that $\hat\rho$ in the algorithm denotes the empirical estimate of the number of relatives.
Our procedure is motivated by the faithfulness assumption, which states that there are no probabilistic independencies other than the ones implied by d-separation \parencite[see e.g.][Section 1.2.3]{pearl2009causality} in the graph. Hence, in linear Gaussian models, zeros in the correlation matrix indicate the absence of an open path between two variables in the sample limit \parencite[][Theorem 3.2]{spirtes2001causation}, meaning that they are not relatives. 
In practice, we may over- or under-estimate the true number of relatives depending on the threshold $\varepsilon$.
For finite samples, thresholding the empirical correlation matrix therefore yields noisy estimates of the number of relatives. 
In our forthcoming experiments we choose $\varepsilon$ to be the minimum absolute sample correlation required for a given significance level assuming Gaussian distributions.

\paragraph{Estimating a DAG given an order using sparse regression.}
To estimate a DAG based on an estimated ordering by the number of relatives, we follow \textcite[Section 3.2]{reisach2023scale}, and rely on ordering-based search \parencite{Teyssier2012Jul}. Given an order, we regress the observations of each variable on those of the variables preceding it  
using sparse regression. We take nonzero regression coefficients to be adjacencies in the estimated adjacency matrix. This yields \cref{algo:regress} (the penalty parameter $\lambda$ is chosen via the Bayesian information criterion).

\begin{algorithm}[b]
    \setlength{\baselineskip}{1.2\baselineskip}  %
    \KwIn{data $(\X_1,\dots,\X_n)$, order $s\in\N^n$}%
    Let $A\in\{0\}^{n\times n}$ be an empty adjacency matrix.\\
    Standardize $(\X_1,\dots,\X_n)$.\\
    \For{$j\in 2,\dots,n$}{
        Let $t=s_j$ be the index of the target variable in the ordering.\\
        Let $S=\{s_1,\dots,s_{j-1}\}$, so $\X_S=(\X_i)_{i\in S}$ are observations of the predecessors.\\
        Estimate the least-squares coefficients $\hat{\beta}^\text{OLS}=\argmin_{\beta\in\R^{j-1}}\|\X_t-\X_S\beta\|_2^2$.\\
        Set $A_{S_\ell t}$ to $1$ for $\ell\in(1,\dots,j-1)$ if $\hat\beta_\ell\neq0$,
        where $\hat\beta=\argmin_{\beta\in\R^{j-1}}\|\X_{t}-\X_S\beta\|_2^2+\lambda_\text{BIC}\sum_{i=1}^{j-1}\frac{|\beta_i|}{|\hat{\beta}_i^{\text{OLS}}|}$.
    } 
    \textbf{Return} $A$
    \caption{Estimating adjacencies given an ordering.}
    \label{algo:regress}
\end{algorithm}

\subsection{Evaluation of sortability by true and estimated relatives across data generation regimes}\label{sec:schemes}
By virtue of being a graphical property, rel-sortability is unaffected by the parameters of the structural equations compatible with a DAG, which have been the focus of the literature aimed at addressing other sortabilities.
To illustrate this aspect, we evaluate sortability by relatives as per \cref{eq:sortability} across different data generation regimes. 
Specifically, we compute i) an oracle rel-sortability based on the true number of relatives, and ii) an empirical rel-sortability based on estimating relatives with \cref{algo:est_rels}.
We compare them to 
sortability by variance \parencite[\textit{var-sortability},][]{reisach2021beware}, and 
sortability by the coefficient of determination $R^2$ \parencite[\textit{$R^2$-sortability},][]{reisach2023scale}.
We generate synthetic data from structural causal models \parencite[\textit{SCMs},][Definition 7.1.1]{pearl2009causality} with linear functions and additive noise \parencite[cf.][]{hoyer2008nonlinear}.
We draw graph structures and parameters independently from one another.
The graph structures are drawn as ER and SF random DAGs with $n=30$ nodes and edge density parameter $c\in(1,\dots,10)$ controlling the expected number of edges per node.
For each combination, we independently draw $10$ graphs. Figures show the mean over these as lines and its $95\%$ bootstrap confidence interval as shaded areas.
We draw observations from 
\begin{itemize}
    \item standardized SCMs \parencite[\textit{sSCMs},][]{reisach2021beware}, 
    \item internally standardized SCMs \parencite[\textit{iSCMs},][]{ormaniec2024standardizing}, 
    \item the DAG-adaptation of the Onion method \parencite[\textit{DaO},][]{andrews2024better}, and
    \item unitless unrestricted Markov-consistent SCM generation \parencite[\textit{UUMC},][]{herman2025unitless}.
\end{itemize}
For each sampled graph, we first parameterize an SCM by drawing slope coefficients independently from $\text{Unif}(0.5,1)$ and drawing independent additive Gaussian noise with standard deviations drawn independently from $\text{Unif}(0.5,1)$. Note that the DAG underlying the SCMs with Gaussian noise and non-equal variances is not identifiable in general. Good recovery of the true DAG in this setting hence suggests that there is additional exploitable structure beyond Markov equivalence. 
The sSCM and iSCM procedures adapt this set of parameters, whereas DaO and UUMC sample their own parameters given the graph. In all cases, the implied parameters of a regular SCM differ between procedures.
For each pair of graph and parameters we draw $m=1000$ observations independently.
Regarding the estimation of the number of relatives in \cref{algo:est_rels}, note that for two independent Gaussians with sample correlation $C_{ij}$, the transformed sample correlation $C_{ij}((m-2)/(1-C_{ij}^2))^{1/2}$ follows a Student $t$-distribution with $m-2$ degrees of freedom.
Hence, we set the threshold $\varepsilon$ to the minimum absolute sample correlation required for a significance level $\alpha=0.05$ under this test.

\paragraph{Results.}
The results are shown in \cref{fig:schemes_ER}; SF DAGs are shown in 
\cref{app:SF}, non-standardized SCMs in 
\cref{app:not-standardized}.
As expected from our simulations in \cref{sec:numerical_prevalence}, oracle rel-sortability rises quickly as $c$ grows, and then falls slowly toward the chance-level value of $1/2$.
The empirical approximation of the relatives is close for sSCMs and iSCMs. For DaO and UUMC the empirical version underestimates the oracle for small $c$, and slightly overestimates it for larger $c$. This indicates that our estimation procedure for the relatives can pick up on other data properties, including possibly numerical artifacts. 
All four schemes control sortability by variance at around $1/2$ by design, in stark contrast to non-standardized SCM data (\cref{fig:non_std_schemes}).
The increase in $R^2$-sortability for denser graphs in sSCMs is known from \cite{reisach2023scale}. The high $R^2$-sortability values for iSCMs add nuance to the claim in \cite{ormaniec2024standardizing} that iSCMs are mostly not $R^2$-sortable (even more so for SF DAGs, see \cref{fig:schemes_SF}).
Most importantly for us, rel-sortability tends to take values well above $1/2$ across graph types, and tends to be highest when $R^2$-sortability is low. This signals a new, likely unintended, and potentially unrepresentative (depending on the domain, cf.\ \cref{tab:real_world_rel_stb}) property of common sampling schemes for random causal DAGs that may be used for causal discovery.
\begin{figure}[H]
    \centering
    \begin{subfigure}{.24\linewidth}
        \centering
        \includegraphics[width=\linewidth]{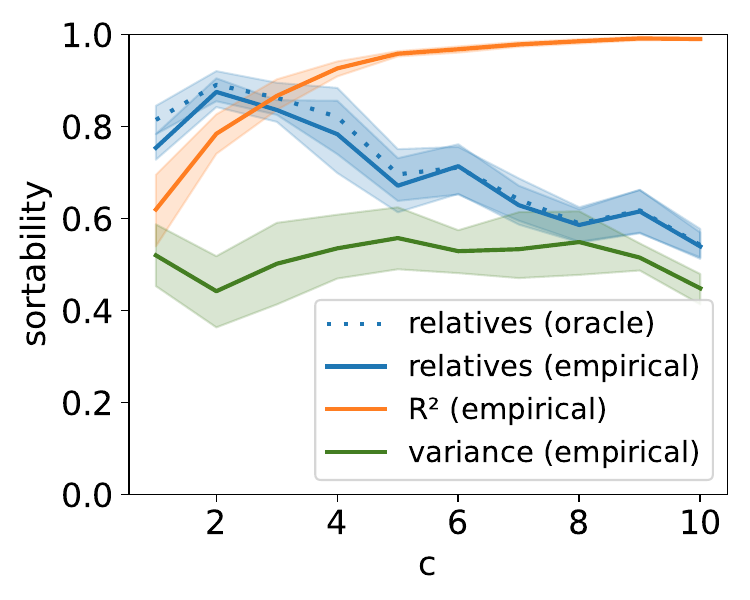}
        \vskip-.5em
        \caption*{sSCM}
    \end{subfigure}
    \begin{subfigure}{.24\linewidth}
        \centering
        \includegraphics[width=\linewidth]{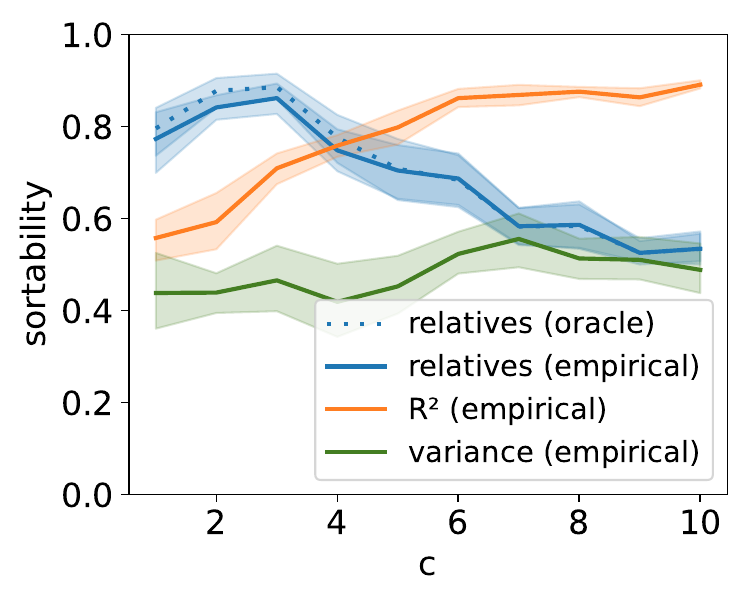}
        \vskip-.5em
        \caption*{iSCM}
    \end{subfigure}
    \begin{subfigure}{.24\linewidth}
        \centering
        \includegraphics[width=\linewidth]{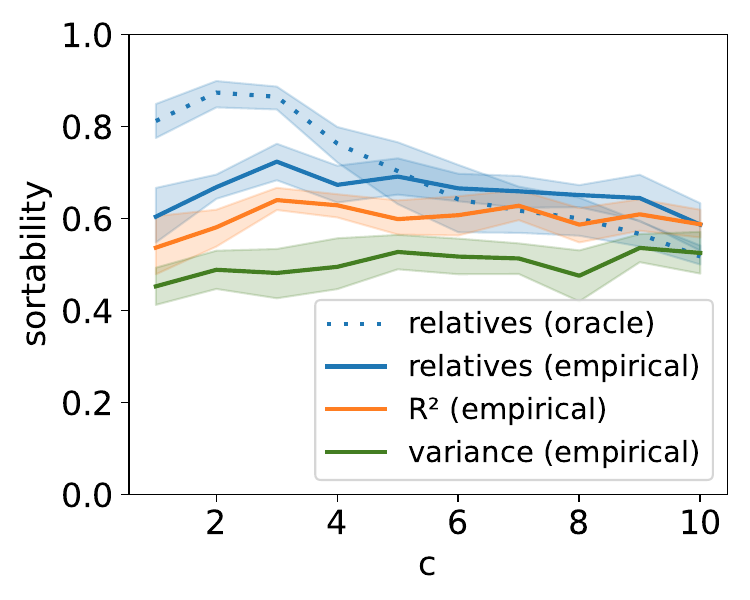}
        \vskip-.5em
        \caption*{DaO}
    \end{subfigure}
    \begin{subfigure}{.24\linewidth}
        \centering
        \includegraphics[width=\linewidth]{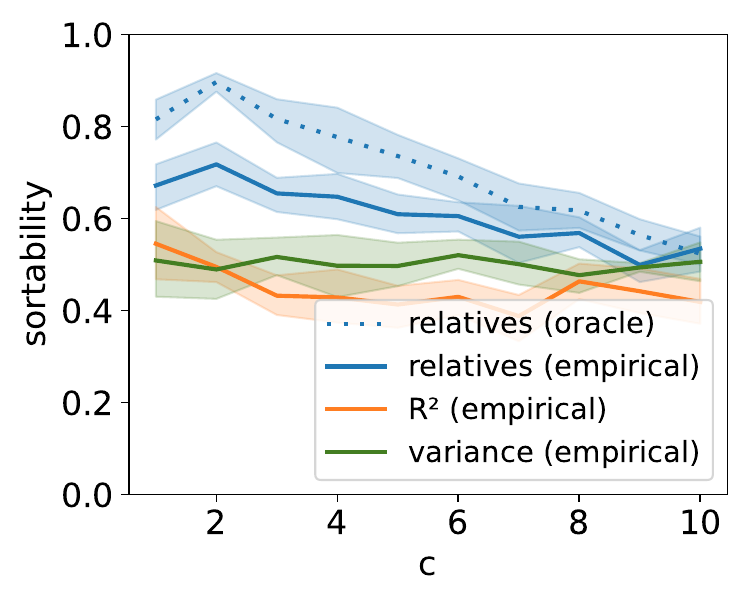}
        \vskip-.5em
        \caption*{UUMC}
    \end{subfigure}
    \caption{Comparison of rel-sortability to other sortabilities in ER DAGs.}
    \label{fig:schemes_ER}
\end{figure}

\subsection{Causal discovery by estimating relatives}\label{sec:emp_causal_discovery}
We now demonstrate the potential impact of rel-sortability for causal discovery.
Combining \cref{algo:est_rels} and \cref{algo:regress} yields a causal discovery algorithm that sorts by the estimated number of relatives and regresses in the resulting order. We term this algorithm \textit{rel-SortnRegress}. For reference, we run the DAG learning algorithm \textit{DAGMA} \parencite{bello2022dagma}, which has been found to perform well on synthetic data from non-standardized SCMs. We choose the L2 loss and default parameters. We also include \textit{FLOP} \parencite{wienobst2025embracing}, the state-of-the-art CPDAG discovery algorithm (recall that the CPDAG of a DAG represents its MEC via undirected edges). We run FLOP with a penalty parameter of $2$ (a typical value) and $10$ restarts.
Finally, we also show $R^2$-\textit{SortnRegress} \parencite{reisach2023scale}, an algorithm of the same kind as rel-SortnRegress based on a sortability criterion that measures predictability of each node by all others.
We draw random graphs and observations as in \cref{sec:schemes}.
To evaluate the quality of a graph learnt by these algorithms, we compare it to the true DAG using the parent-based adjustment identification distance \parencite[\textit{parent-AID},][Section 4.1]{henckel2024adjustment}. 
Given a learnt DAG, parent-AID assesses how often it would lead to incorrect interventional conclusions when using the parents of a node as adjustment set. Given a learnt CPDAG, it counts not-identifiable pairs as mistakes relative to the true DAG. This makes parent-AID a causally meaningful evaluation criterion for DAGs and CPDAGs. We additionally show the structural Hamming distance \parencite[\textit{SHD}, see e.g.][Definition 1]{peters2015structural}, which counts each node pair whose edge differs between the true DAG and a learnt graph, and which is a widely used and easily interpretable graph distance.
We show a normalized version of both distances following \textcite{henckel2024adjustment}.

\paragraph{Results.}
We show results for ER DAGs in terms of parent-AID (\cref{fig:ER_pAID}) and SHD (\cref{fig:ER_SHD}). Results for SF DAGs can be found in \cref{app:SF}, and results on non-standardized SCM data in \cref{app:not-standardized}.
Strikingly, we observe that FLOP outperforms all other algorithms in all but the densest graphs, across all settings (we investigate this further in \cref{sec:MEC}).
Our algorithm rel-SortnRegress yields remarkably good results, coming close to FLOP on SCM, sSCM, and iSCM for $c\in\{1,2\}$, a most common choice in the literature (cf.\ \cref{tab:literature}). It also outperforms or matches $R^2$-SortnRegress for moderately sparse graphs in iSCM, DaO, and UUMC, the three data generation regimes specifically designed to counteract other sortability patterns.
This corroborates our theoretical analysis of the relevance of relatives to causal discovery, and suggests its utility as a baseline for assessing the simplicity of the causal discovery due to topological patterns of random DAG generation schemes.
Also notably, DAGMA yields competitive results only on non-standardized settings, suggesting that it relies on information in the marginal variances, which is problematic since real-world data scales are typically arbitrary \parencite[cf.][]{reisach2021beware}. Its performance degrades more gracefully in SHD, indicating that its sparsity may be well-attuned to the setting despite struggling with order recovery.
\begin{figure}[H]
    \centering
    \begin{subfigure}{.24\linewidth}
        \centering
        \includegraphics[width=\linewidth]{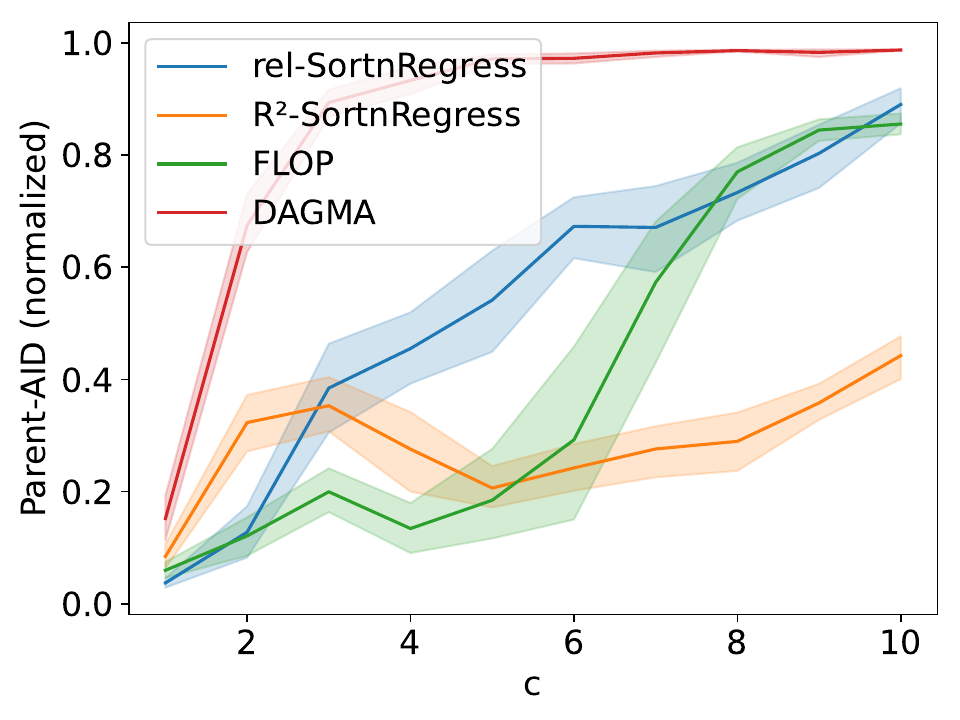}
        \vskip-.5em
        \caption*{sSCM}
    \end{subfigure}
    \begin{subfigure}{.24\linewidth}
        \centering
        \includegraphics[width=\linewidth]{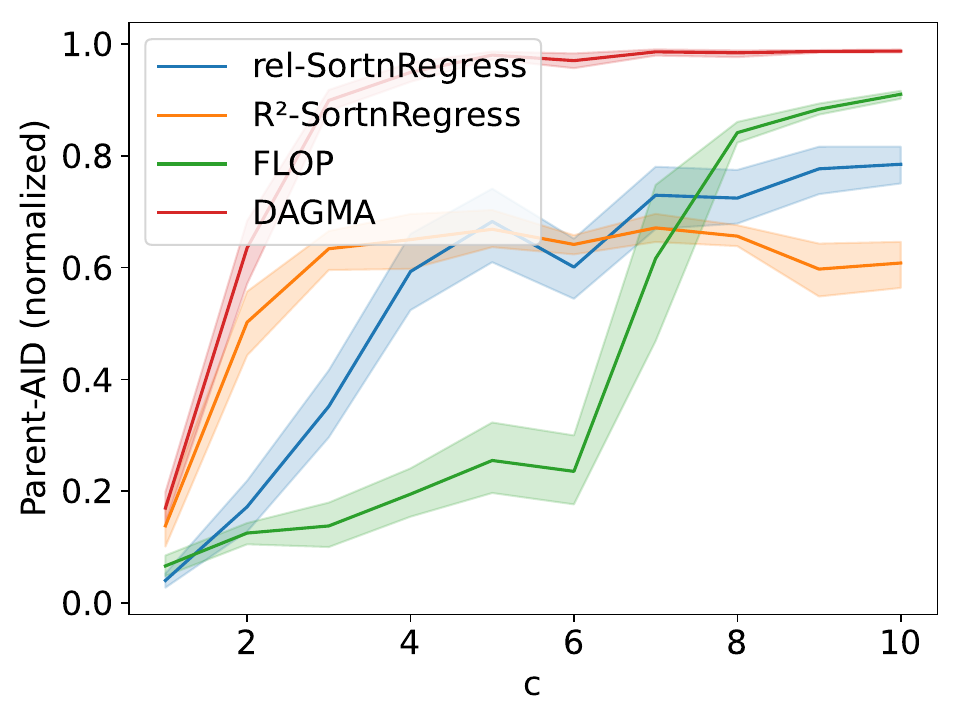}
        \vskip-.5em
        \caption*{iSCM}
    \end{subfigure}
    \begin{subfigure}{.24\linewidth}
        \centering
        \includegraphics[width=\linewidth]{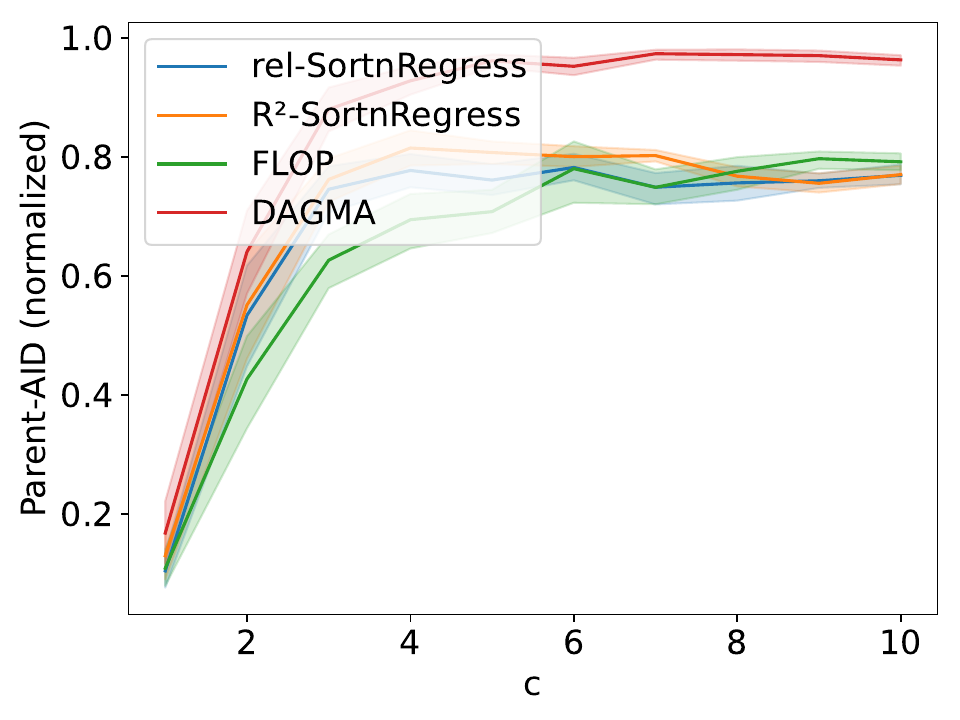}
        \vskip-.5em
        \caption*{DaO}
    \end{subfigure}
    \begin{subfigure}{.24\linewidth}
        \centering
        \includegraphics[width=\linewidth]{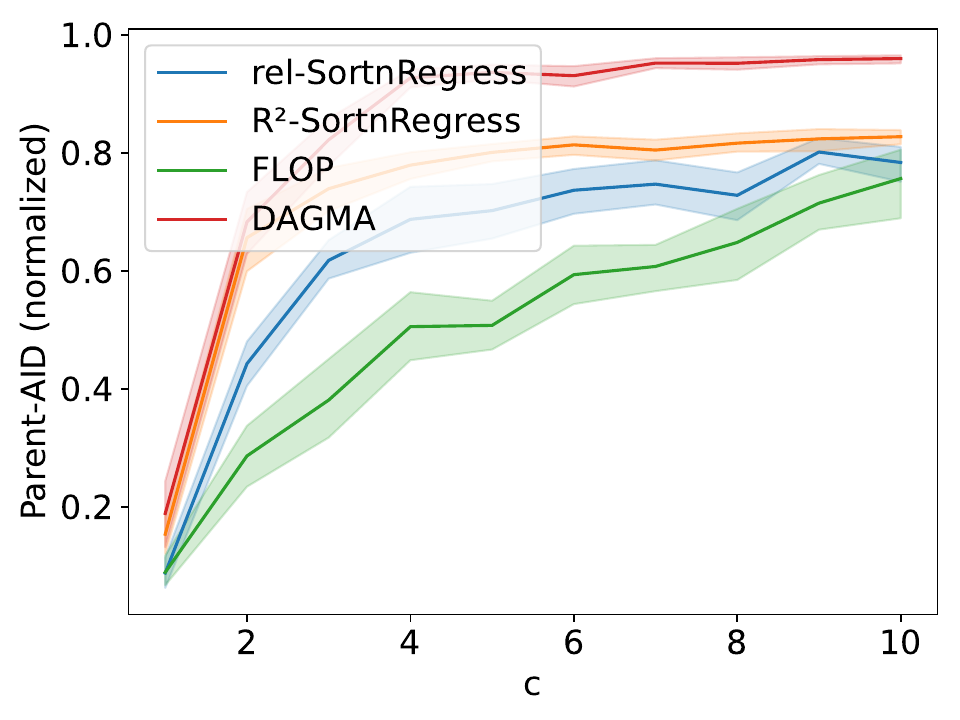}
        \vskip-.5em
        \caption*{UUMC}
    \end{subfigure}
    \caption{Causal discovery results for ER DAGs (parent-AID to true DAG, lower is better).}
    \label{fig:ER_pAID}
\end{figure}
\begin{figure}[H]
    \centering
    \begin{subfigure}{.24\linewidth}
        \centering
        \includegraphics[width=\linewidth]{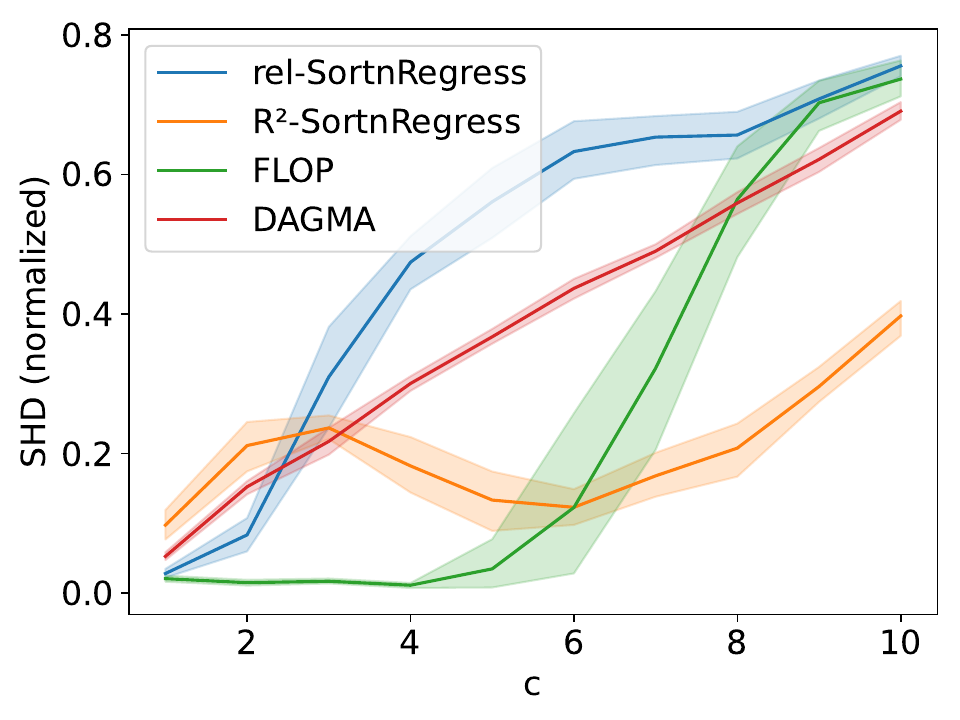}
        \vskip-.5em
        \caption*{sSCM}
    \end{subfigure}
    \begin{subfigure}{.24\linewidth}
        \centering
        \includegraphics[width=\linewidth]{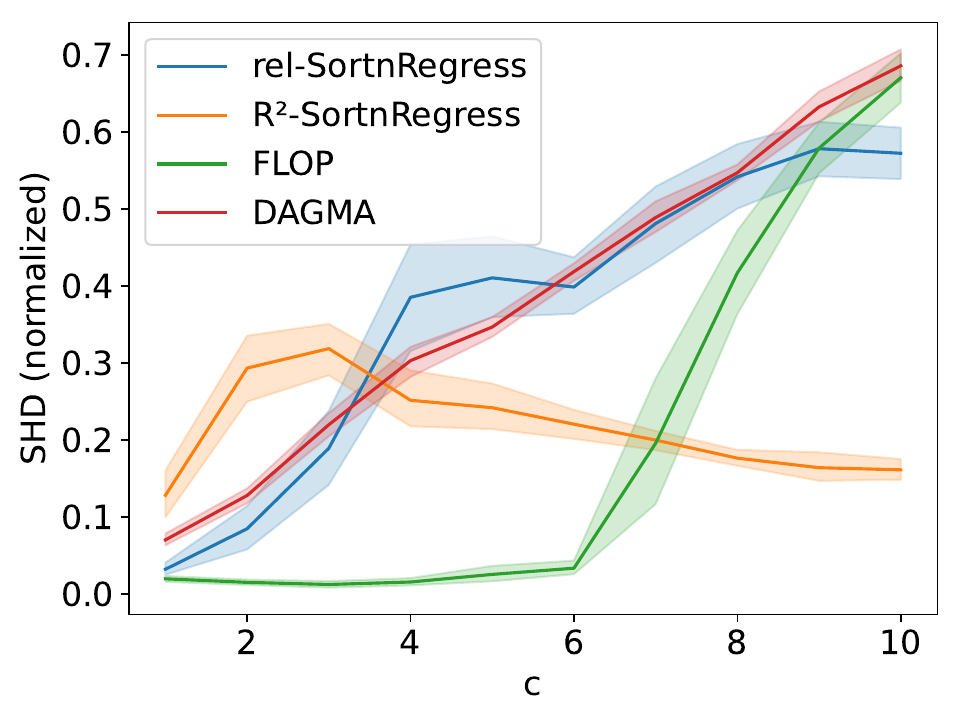}
        \vskip-.5em
        \caption*{iSCM}
    \end{subfigure}
    \begin{subfigure}{.24\linewidth}
        \centering
        \includegraphics[width=\linewidth]{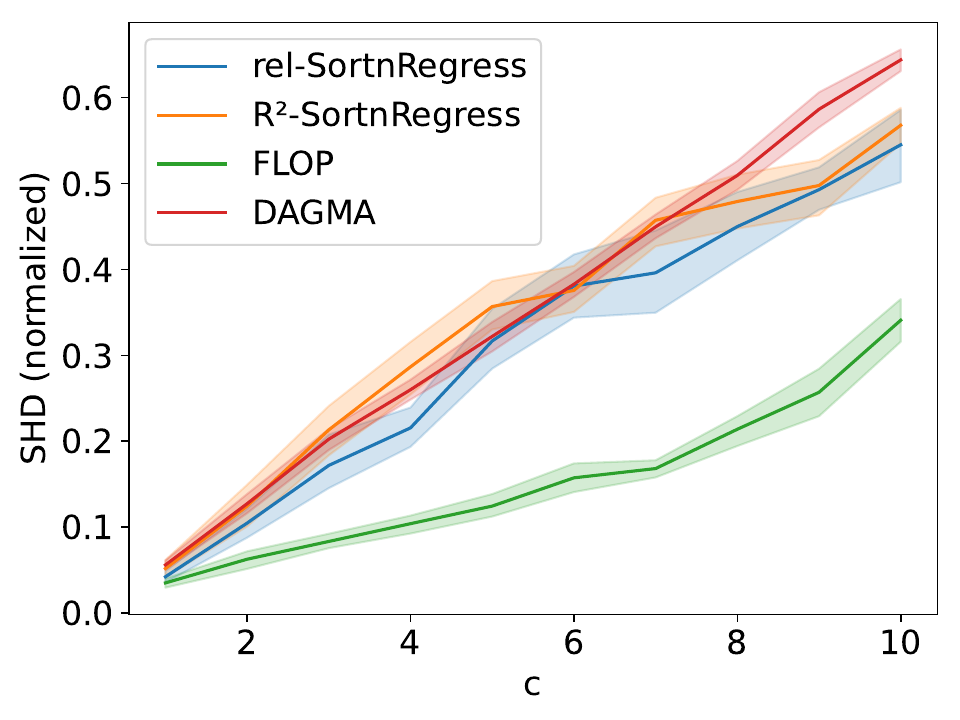}
        \vskip-.5em
        \caption*{DaO}
    \end{subfigure}
    \begin{subfigure}{.24\linewidth}
        \centering
        \includegraphics[width=\linewidth]{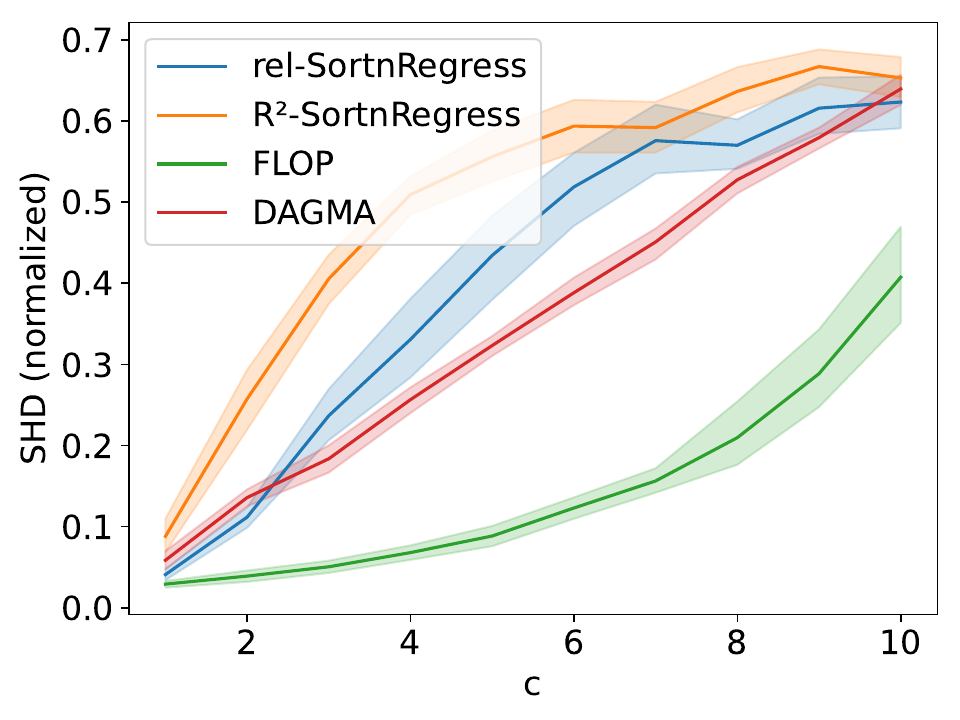}
        \vskip-.5em
        \caption*{UUMC}
    \end{subfigure}
    \caption{Causal discovery results for ER DAGs (SHD to true DAG, lower is better).}
    \label{fig:ER_SHD}
\end{figure}

%% file: chapters/MEC_size.tex
\section{Connecting rel-sortability and MEC size}\label{sec:MEC}
\paragraph{FLOP performance indicates small MECs.}
The excellent performance of FLOP in our causal discovery experiments in \cref{sec:emp_causal_discovery} suggests that the MECs of the simulated DAGs are small, which may be surprising since the MECs for sparse graphs can be large \parencite{jahn2025lower}.
To test this hypothesis, we evaluate the size of the MECs for our experimental settings empirically, as detailed in \cref{app:mec_sizes}.
For ER and SF DAGs, we find that MECs tend to be small and the fraction of undirected edges in the CPDAG of the true DAG tends to be low for $c>1$.
In light of this, we conjecture that CPDAG algorithms may have considerable untapped potential as a first stage of DAG discovery that can direct most edges.
That said, we caution that it is possible that FLOP's performance is, at least in part, due to data artifacts yet unknown, and these results call for further investigation.

\paragraph{Strict rel-sortability implies singleton MECs.}
Prompted by our observation that MECs tend to be small in our simulations, we investigate the connection between high rel-sortability and the size of the MEC. 
Recall from \cref{sec:numerical_prevalence} that sparse ER DAGs (provided $c\geq1$) and sparse SF DAGs are characterized by high rel-sortability.
The MEC of a DAG with high rel-sortability may be expected to be small, since the differences in relatives could result in unshielded colliders which can be oriented by the rules in \textcite{meek1995orientation}. 
We formalize this idea in \cref{thm:cpdag_dag}.
Thus, rel-sortability being high may, at least partially, explain why FLOP performs well, even in sparse settings.
\cref{thm:cpdag_dag} links rel-sortability to CPDAG discovery, and raises the question to what extent the evaluation of CPDAG discovery algorithms more broadly may also be influenced by potentially unrepresentative properties of synthetic data generation schemes.

\begin{theorem}
    For a DAG $\G=(\V,\E)$ with $\s_\rho(\G)=1$, the MEC of $\G$ contains only $\G$ itself.
    \label{thm:cpdag_dag}
\end{theorem}
\noindent
The proof of \cref{thm:cpdag_dag} (see \cref{app:proof_thm3}) is based upon showing that all edges in the skeleton of a DAG $\G$ with $\s_\rho(\G)=1$ can be oriented by the orientation rules in \cite{meek1995orientation}.
We note for completeness that this result establishes the relationship between rel-sortability and MEC size in the extreme case of a singleton MEC, whereas the general relationship between the two remains open.

%% file: chapters/time_DAGs.tex
\section{Time-unrolling ER and SF graphs into DAGs}\label{sec:time_series}
Our theoretical and numerical analyses show that rel-sortability is high for ER and SF DAGs generated with common parameters and using the usual fashion of directing edges in undirected graphs.
As observed in our experiments, this can impact the difficulty of the causal discovery task.
Hence, it appears that synthetic data generation mechanisms do not reflect the diversity of real-world data as observed for our examples in \cref{tab:real_world_rel_stb}, and evaluations may thus misestimate the promise of causal discovery algorithms in applications \parencite[cf.][]{brouillard2025landscape}.
To address this problem, we propose turning to graph sampling schemes based on DAG models for time series as a possible alternative to tuning ER and SF DAG sampling parameters. Time series DAGs have a clear motivation in the applied literature and real-world benchmarking efforts. For recent examples, see e.g. \cite{bang2023we}, \cite{gobler2024texttt}, and \cite{gamella2025causal}, all of which note and use the fact that real-world measurements are taken at specific time points. Moreover, \cite{reisach2026time} argue that an observation of a causal system is necessarily a time series provided that causes precede effects, regardless of whether observations are i.i.d.
Hence, a time order can justify acyclicity in the underlying graph without being supplied as background knowledge, and time-unrolled DAGs can be passed to standard causal discovery methods just like any other DAG.
Existing time series sampling schemes such as \cite{cheng2023causaltime} and \textcite[][Section 6]{herman2025unitless} can be used to generate i.i.d.\ observations by repeatedly and independently simulating short time series trajectories. 
Thus, when using undirected random graph models such as ER and SF as a basis, unrolling them in time 
may be a natural alternative to orienting edges along a vertex order to obtain a DAG.

\paragraph{Constructing time-unrolled DAGs and controlling rel-sortability.}
We show how a potentially cyclic graph can be unrolled in time, and state a condition to control rel-sortability in the resulting DAG. Our proposal is closely related to the ideas of a \textit{summary graph} and \textit{full-time graph} in the time series causality literature \parencites[cf.][]{peters2013causal}[][Section 1.3]{assaad2021time}.
Let $\G_S=(\V_S,\E_S)$ be a non-empty directed (possibly cyclic) graph with $p\geq2$ nodes, which we refer to as the summary graph in analogy to the convention in time series causality. We assume $\G_S$ contains a self-loop for every node, that is, $(v,v)\in\E_S$ for every $v\in\V_S$. 
For an integer time horizon $T>1$, define the (finite) time-unrolled DAG $\mathcal G_T=(\mathcal V_T,\mathcal E_T)$ with
\begin{align*}
    \mathcal V_T&=\{v^{(t)} : v\in\V_S,\ t\in\{1,\ldots,T\}\},\\
    \mathcal E_T&= \bigl\{(i^{(t)},j^{(t+1)}) \colon (i,j)\in\mathcal E_S,\ t\in\{1,\ldots,T-1\}\bigr\}.
\end{align*}
Random DAGs can then be sampled by first drawing $\G_S$, choosing $T$, and then constructing $\G_T$.
\cref{fig:ts} provides an example of a summary graph and its time-unrolled DAG for $T=4$.
\begin{figure}[H]
    \centering
    \begin{subfigure}{.29\linewidth}
        \centering
        \begin{tikzpicture}[scale=0.8]
            \node (a) at (0,0) {$a$};
            \node (ap) at (-0.1,0.1) {$\phantom{a}$};
            \node (b) at (1.5,1.2) {$b$};
            \node (bp) at (1.4,1.3) {$\phantom{b}$};
            \node (c) at (3,0) {$c$};
            \draw[->] (a)--(b);
            \draw[->] (bp)--(ap);
            \draw[->] (b)--(c);
            \draw[->] (c)--(a);
            \draw[->] (a)edge[loop left](a);
            \draw[->] (b)edge[loop above](b);
            \draw[->] (c)edge[loop right](c);
        \end{tikzpicture}
        \caption{Summary graph.}
        \label{fig:summary_graph}
    \end{subfigure}
    \begin{subfigure}{.49\linewidth}
        \centering
        \begin{tikzpicture}[scale=0.6]
            \pgfmathsetmacro{\dzero}{0};
            \pgfmathsetmacro{\done}{-1};
            \pgfmathsetmacro{\dtwo}{-2};
            \node[] (a1) at (0,\dzero) {$a^{(1)}$};
            \node[] (a2) at (2.5,\dzero) {$a^{(2)}$};
            \node[] (a3) at (5,\dzero) {$a^{(3)}$};
            \node[] (a4) at (7.5,\dzero) {$a^{(4)}$};
            \node[] (b1) at (0,\done) {$b^{(1)}$};
            \node[] (b2) at (2.5,\done) {$b^{(2)}$};
            \node[] (b3) at (5,\done) {$b^{(3)}$};
            \node[] (b4) at (7.5,\done) {$b^{(4)}$};
            \node[] (c1) at (0,\dtwo) {$c^{(1)}$};
            \node[] (c2) at (2.5,\dtwo) {$c^{(2)}$};
            \node[] (c3) at (5,\dtwo) {$c^{(3)}$};
            \node[] (c4) at (7.5,\dtwo) {$c^{(4)}$};

            \draw[->] (a1)--(a2);
            \draw[->] (b1)--(b2);
            \draw[->] (c1)--(c2);
            \draw[->] (a2)--(a3);
            \draw[->] (b2)--(b3);
            \draw[->] (c2)--(c3);
            \draw[->] (a3)--(a4);
            \draw[->] (b3)--(b4);
            \draw[->] (c3)--(c4);
            \draw[->] (a1)--(b2);
            \draw[->] (b1)--(a2);
            \draw[->] (b1)--(c2);
            \draw[->] (c1)--(a2);
            \draw[->] (a2)--(b3);
            \draw[->] (b2)--(a3);
            \draw[->] (b2)--(c3);
            \draw[->] (c2)--(a3);
            \draw[->] (a3)--(b4);
            \draw[->] (b3)--(a4);
            \draw[->] (b3)--(c4);
            \draw[->] (c3)--(a4);
        \end{tikzpicture}
        \caption{Time-unrolled DAG.}
        \label{fig:ts_dag}
    \end{subfigure}
    \caption{A directed and cyclic summary graph and corresponding time-unrolled DAG for $T=4$.}
    \label{fig:ts}
\end{figure}
To see how rel-sortability can be controlled in time-unrolled DAGs, observe first that in \cref{fig:summary_graph}, there is a directed path between all node pairs. Note further that, as a consequence of this and the self-loops, all nodes with a time index $\geq3$ in \cref{fig:ts_dag} have all nodes of the graph as relatives, meaning they have the same number of relatives.
Based on this idea, \cref{thm:ts_dags} states a sufficient condition on the summary graph such that the resulting time series DAG tends to a rel-sortability of $1/2$ as $T$ goes to infinity.
The rel-sortability in a time-unrolled DAG can thus be controlled by choosing a summary graph that meets the condition in \cref{thm:ts_dags} and a sufficiently long time horizon.
\begin{theorem}
    \label{thm:ts_dags}
    For $i,j\in\V_S$, let $i\leadsto j$ denote the existence of a directed path from $i$ to $j$.
    It holds that 
    \begin{equation}
        \forall i,j\in\V_S,(i\leadsto j \implies j\leadsto i) \implies \lim_{T\to\infty}\s_\rho(\G_T) = 1/2.
        \label{eq:implication}
    \end{equation}
\end{theorem}
The proof is shown in \cref{app:proof_thm4}.
The condition in \cref{eq:implication} states that every weakly connected component of the directed summary graph is strongly connected. Since we assume self-loops at every node, the adjacency matrix of each component, viewed as the support of a Markov transition matrix, is irreducible and aperiodic.
Sampling a graph that meets this condition is not difficult -- for example, it is straightforwardly fulfilled by symmetry when sampling an \ERname graph and replacing each undirected edge by two directed ones in opposing directions. That said, a graph need not be symmetric to meet the condition, as can be seen by the summary graph used for our example in \cref{fig:summary_graph}. 
To provide a numerical illustration of \cref{thm:ts_dags}, we show the oracle rel-sortability for time-series graphs with different time-horizons based on randomly drawn summary graphs in \cref{app:ts_dags}. 
As expected from \cref{thm:relatives} and \cref{thm:ts_dags}, we observe that rel-sortability approaches $1/2$ from above with increasing time horizon.

%% file: chapters/discussion.tex
\section{Discussion}\label{sec:discussion}

\paragraph{Limitations.}
We acknowledge several limitations to our theoretical and experimental results.
Our lower bound in \cref{thm:lower_bound} only covers ER DAGs, and though our experiments suggest a similar trend for SF DAGs, we do not have a corresponding theoretical result. 
Regarding \cref{thm:cpdag_dag}, we note that it does not imply that highly rel-sortable DAGs have small MECs in general, and a result to this end would present a valuable extension.
In our experiments, we set the threshold for estimating the number of relatives based on the assumption of Gaussianity, matching the data generation. The practical performance may change for different significance levels or when using nonparametric threshold selection.
In our causal discovery experiments we do not tune or search for optimal hyperparameters of any of the algorithms. 
In addition, all our experiments use linear models, and although we use a range of different data generation regimes, we do so on the basis of parameters drawn from a particular set of distributions. Additional experiments or assumptions would be needed to extend our empirical observations 
to a wider range of SCMs, in particular ones with nonlinear functions. Finally, we propose time-unrolled DAGs as a potential alternative to existing graph sampling schemes and show how their rel-sortability can be controlled. Our theoretical result is specific to summary graphs in which every weakly connected component is strongly connected, and our empirical analysis is limited to a numerical illustration for finite time horizons.
Further research is needed to judge the suitability of time-series DAGs in the context of causal discovery more comprehensively.

\paragraph{Opportunities for causal discovery.}
Our contribution introduces the number of relatives of a node as a new aspect in the importance of graph topology for causal discovery. 
The connection between sortability by the number of relatives and the size of the MEC outlined in our work may open the door for new theoretical identifiability results based on topological properties that restrict the size of the MEC.
Our work also presents many opportunities to improve the use of relatives for causal discovery. For example, it may be interesting to choose the threshold $\varepsilon$ in \cref{algo:est_rels} without relying on a parametric assumption, or to choose a different threshold adaptively for different variables. In addition, one could make use of the fact that the set of relatives of a given node must contain the set of relatives of all its ancestors, an aspect we do not yet use in our algorithms. It would also be possible to use the ordering by estimated relatives as an initialization or preprocessing step for other algorithms.
Finally, the different ways of defining sortability in the literature, including our own edge-wise one, point to a broader question of how to compare partial and weak orders that may be relevant beyond causal discovery.

\paragraph{Evaluating and training algorithms on synthetic data.}
Evaluations on synthetic data serve as a controlled testbed for causal discovery algorithms.
We highlight that predominantly sampling random DAGs with high rel-sortability may be problematic when evaluating causal discovery algorithms, since it may not reflect the diversity of rel-sortability in real-world data. Hence, the currently dominating sampling regimes in the literature risk overfitting to a particular and potentially unrepresentative graph class.
This possibility is especially acute when using supervised learning, as for example in \cite{ke2023learning}, and may also be of interest for tabular foundation models \parencite[cf.][]{hollmann2025accurate}, which often rely on similar data generation schemes for pretraining.
Our findings call for a systematic evaluation of which graph sampling scheme may be considered realistic, and suggest that a reevaluation of existing causal discovery algorithms may be needed.
Our result on controlling rel-sortability in certain types of time-unrolled DAGs constitutes a first step in the direction of alternative graph sampling schemes. We consider further investigations into the topological properties of such DAGs and the question of how to parameterize the resulting SCMs two areas of high interest and consequence for future research.
Prospectively, sampling data from independent trajectories of short time-unrolled DAGs also raises the question to what extent full or partial knowledge of the time order can and should be used for causal discovery on i.i.d.\ data.

%% file: chapters/conclusion.tex
\section{Conclusion}\label{sec:conclusion}

Our contribution provides new insight into the topological properties of random DAG models in causal discovery. 
We show that the set of nodes reachable via open paths, termed relatives, never decreases along the causal order, and tends to increase strictly in common random DAG models.
We prove a lower bound for ER DAGs that indicates high sortability by relatives for common density regimes in the literature, and perform a numerical analysis that shows a similar pattern for SF DAGs.
We propose an algorithm for estimating the number of relatives empirically to obtain a causal order, and our causal discovery experiments illustrate that the property is highly relevant for causal discovery. 
This is potentially problematic, since the rel-sortability of benchmark settings may not represent the diversity of real-world datasets.
In addition, we show that a strict increase in the number of relatives along the causal order implies a singleton MEC. Empirically we find the MECs to be small for our settings, suggesting that MEC discovery methods may have untapped potential as a first stage of DAG discovery that directs most edges.
Inspired by recent examples in applied causal discovery, we propose unrolling random graphs in time as an alternative way to sample random DAGs, and give a sufficient condition under which their rel-sortability can be controlled.
Overall, our contribution renews questions about the real-world performance of existing causal discovery algorithms, and provides rel-sortability as a new diagnostic for measuring the diversity of synthetic data from causal models, as well as a potential tool for causal discovery.

%% file: chapters/appendix.tex
\section{Illustration of Relatives}\label{app:relatives}

\tikzset{
    exnode/.style={circle, draw, minimum size=5.5mm, inner sep=0pt},
    hl/.style={fill=urlcolor!25},
    faded/.style={black!35},
    na/.style={}, nb/.style={}, nc/.style={}, nd/.style={}, ne/.style={}, nf/.style={},
    eae/.style={}, eab/.style={}, ebd/.style={}, ecd/.style={}, ecf/.style={},
    relnum/.style={font=\footnotesize, text=citecolor},
}
\newcommand{\exampledag}{%
    \node[exnode, na] (a) at (0, 1.6) {$a$};
    \node[exnode, nc] (c) at (2.8, 1.6) {$c$};
    \node[exnode, ne] (e) at (-0.9, 0.6) {$e$};
    \node[exnode, nb] (b) at (0.9, 0.6) {$b$};
    \node[exnode, nf] (f) at (3.7, 0.6) {$f$};
    \node[exnode, nd] (d) at (1.9, -0.4) {$d$};
    \draw[->, eae] (a) -- (e);
    \draw[->, eab] (a) -- (b);
    \draw[->, ebd] (b) -- (d);
    \draw[->, ecd] (c) -- (d);
    \draw[->, ecf] (c) -- (f);
}

We illustrate the definition of relatives in \cref{def:relatives} using a small example DAG shown in \cref{fig:relatives_example}.
The ancestors of node $b$ are $\anc(b)=\{a,b\}$, see \cref{fig:relatives_anc}. The relatives of $b$ are the union of the descendants of these ancestors, $\rel(b)=\desc(a)\cup\desc(b)=\{a,b,d,e\}\cup\{b,d\}=\{a,b,d,e\}$, see \cref{fig:relatives_rel}.
Equivalently, $x\in\rel(y)$ if and only if $x$ and $y$ have a common ancestor. The nodes $c$ and $f$ have no common ancestor with $b$ and are therefore not relatives of $b$.
\cref{fig:relatives_count} shows the number of relatives of every node. In accordance with \cref{thm:relatives}, it is non-decreasing along directed paths.
\begin{figure}[H]
    \centering
    \begin{subfigure}{.32\linewidth}
        \centering
        \begin{tikzpicture}[>=stealth, scale=.7]
            \tikzset{na/.style={hl}, nb/.style={hl, very thick}}
            \exampledag
        \end{tikzpicture}
        \caption{Ancestors of $b$.}
        \label{fig:relatives_anc}
    \end{subfigure}
    \hfill
    \begin{subfigure}{.32\linewidth}
        \centering
        \begin{tikzpicture}[>=stealth, scale=.7]
            \tikzset{na/.style={hl}, nb/.style={hl, very thick}, nd/.style={hl}, ne/.style={hl}, nc/.style={faded}, nf/.style={faded}, ecd/.style={faded}, ecf/.style={faded}}
            \exampledag
        \end{tikzpicture}
        \caption{Relatives of $b$.}
        \label{fig:relatives_rel}
    \end{subfigure}
    \hfill
    \begin{subfigure}{.32\linewidth}
        \centering
        \begin{tikzpicture}[>=stealth, scale=.7]
            \exampledag
            \node[relnum, anchor=south] at (a.north) {4};
            \node[relnum, anchor=south] at (c.north) {3};
            \node[relnum, anchor=east] at (e.west) {4};
            \node[relnum, anchor=south west] at (b.north east) {4};
            \node[relnum, anchor=west] at (f.east) {3};
            \node[relnum, anchor=west] at (d.east) {6};
        \end{tikzpicture}
        \caption{Number of relatives.}
        \label{fig:relatives_count}
    \end{subfigure}
    \caption{Ancestors, relatives, and number of relatives in an example DAG.}
    \label{fig:relatives_example}
\end{figure}

\section{Proof Details}\label{app:proof_details}

This section provides additional details on the proofs of \cref{thm:lower_bound}, \cref{thm:cpdag_dag}, and \cref{thm:ts_dags}.

\subsection{Proof of \cref{thm:lower_bound} (Additional Details)}

This section first states and proves \cref{lem:ancestors}, which is subsequently used in the proof for \cref{thm:lower_bound}.

\subsubsection{\cref{lem:ancestors}}

\begin{lemma}
    \label{lem:ancestors}
    For any $x,y\in\V$,
    \begin{align}
        \rel(y)\setminus\rel(x)\neq\emptyset \iff \anc(y)\setminus\rel(x)\neq\emptyset.
        \label{eq:lemma2_1}
    \end{align}
    Thus, by \cref{thm:relatives}, if $x\in\anc(y)$, then
    \begin{align}
        \rel(x)\subsetneq\rel(y) \iff \anc(y)\setminus\rel(x)\neq\emptyset,
        \label{eq:lemma2_2}
    \end{align}
    and in particular
    \begin{align}
        \rel(x)\subsetneq\rel(y) \iff \roots(y)\setminus\rel(x)\neq\emptyset,
        \label{eq:lemma2_3}
    \end{align}
    and consequently
    \begin{align}
        \rel(x)\subsetneq\rel(y) \iff \roots(y)\setminus\roots(x)\neq\emptyset.
        \label{eq:lemma2_4}
    \end{align}
\end{lemma}
\noindent
\begin{proof}[Proof of \cref{lem:ancestors}]
    We begin by noting that $z\in\rel(x)$ implies $\desc(z)\subseteq\rel(x)$, as follows directly from the definition. 
    Hence,
    \begin{align*}
        \rel(y)\setminus\rel(x) = \desc\left(\anc(y)\right) \setminus \rel(x)
        &= \bigcup_{z\in\anc(y)}\left(\desc(z)\setminus\rel(x)\right)\\
        &= \bigcup_{z\in\anc(y)\setminus\rel(x)}\left(\desc(z)\setminus\rel(x)\right).
    \end{align*}
    In view of this equality, $\rel(y)\setminus\rel(x)\neq\emptyset$ implies $\anc(y)\setminus\rel(x)\neq\emptyset$ because 
    the contrapositive holds trivially. Thus, the implication from left to right in \cref{eq:lemma2_1} holds. 
    
    Assume now that $\anc(y)\setminus\rel(x)\neq\emptyset$, and let $t\in\anc(y)\setminus\rel(x)$. Then $\{t\}\subseteq \desc(t)\setminus\rel(x)$, and therefore
    \begin{align*}
        \{t\} \subseteq \bigcup_{z\in\anc(y)\setminus\rel(x)}\left(\desc(z)\setminus\rel(x)\right)\neq\emptyset.
    \end{align*}
    Thus, the implication from right to left in \cref{eq:lemma2_1} holds.
    By \cref{thm:relatives}, we further know that if $x\in\anc(y)$ then $\rel(x)\subseteq\rel(y)$. Therefore
    \begin{align*}
        \rel(x)\subsetneq\rel(y) \iff \rel(y)\setminus\rel(x)\neq\emptyset \iff \anc(y)\setminus\rel(x)\neq\emptyset,
    \end{align*}
    proving \cref{eq:lemma2_2}.
    We now consider \cref{eq:lemma2_3}. By \cref{eq:lemma2_2}, it suffices to show that
    \begin{align*}  
        \anc(y)\setminus \rel(x)\neq\emptyset
        \quad\Longleftrightarrow\quad
        \roots(y)\setminus \rel(x)\neq\emptyset .
    \end{align*}
    The implication from right to left is immediate since \(\roots(y)\subseteq \anc(y)\).
    For the converse, let \(v\in \anc(y)\setminus \rel(x)\). Since every ancestor of \(y\) has a root ancestor in \(\roots(y)\), there exists \(r\in\roots(y)\) with \(v\in\desc(r)\). If \(r\in\rel(x)\), then \(\desc(r)\subseteq\rel(x)\), and hence \(v\in\rel(x)\), yielding a contradiction. Therefore \(r\notin\rel(x)\), so \(r\in\roots(y)\setminus\rel(x)\).
    This proves \cref{eq:lemma2_3}.

    Finally, \cref{eq:lemma2_4} follows straightforwardly from the fact that among the relatives of $x$, only roots of $x$ can also be roots of $y$.
    This concludes the proof.
\end{proof}

\subsubsection{Proof of \cref{thm:lower_bound}}\label{app:thm_2}

\begin{proof}
    \cref{lem:ancestors} allows us to express the condition for a sortability of $1$ for $(x,y)\in\E$ as
    \begin{align*}
        \s_\rho\big((x,y)\big)=1 \iff \roots(y)\setminus\rel(x)\neq\emptyset.
    \end{align*}
    Thus, the event %
    \begin{align*}
        E\coloneq\big(\{a\in\roots(y)\cap \pa(y)\colon r_x<r_a<r_y\}\neq\emptyset\big) \implies \s_\rho\big((x,y)\big)=1.
    \end{align*}
    By the independence of edges in ER graphs and the definition of $\s_\rho$ it holds that
    \begin{align*}
        \Pr(\s_\rho\big((x,y)\big)=1\mid (x,y)\in \E)
        &\geq \Pr(E).
    \end{align*}
    We have for any $v\in\V\setminus\{x,y\}$ that
    \begin{align*}
        \Pr(v\text{ is a root node}) = (1-p)^{r_v-1}.
    \end{align*}
    The probability of $v\in\pa(y)$ is $p$, independently of whether $v$ is a root node. For distinct $v$, the events $\{v\text{ is a root node and } v\in\pa(y)\}$ depend on disjoint sets of node pairs and are hence independent. Using $1-u\leq e^{-u}$ for the product factors in the inequality and summing the geometric series,
    \begin{align*}
        \Pr(E) &= 1-\prod_{r=r_x+1}^{r_y-1}\left(1-p(1-p)^{r-1}\right)\\
        &\geq 1-\exp\Big(-\sum_{r=r_x+1}^{r_y-1}p(1-p)^{r-1}\Big)\\
        &= 1-\exp\Big((1-p)^{r_y-1}-(1-p)^{r_x}\Big).
    \end{align*}
    Since $p=2c/(n-1)$, $r_x=\lceil q_xn\rceil$, and $r_y=\lceil q_yn\rceil$, we have $(1-p)^{r_x}\to e^{-2cq_x}$ and $(1-p)^{r_y-1}\to e^{-2cq_y}$ as $n\to\infty$, which yields
    \begin{align*}
        \liminf_{n\to\infty}\Pr\big(\s_\rho\big((x,y)\big)=1\mid(x,y)\in\E\big)\geq 1-\exp\left(e^{-2cq_y}-e^{-2cq_x}\right).
    \end{align*}
    Finally, by \cref{thm:relatives}, $\s_\rho\big((x,y)\big)\in\{1/2,1\}$ for $(x,y)\in\E$. Hence
    \begin{align*}
        \liminf_{n\to\infty}\Exp{\s_\rho\big((x,y)\big)\mid(x,y)\in\E}
        &= \frac12+\frac12\liminf_{n\to\infty}\Pr\big(\s_\rho\big((x,y)\big)=1\mid(x,y)\in\E\big)\\
        &\geq \frac12+\frac12\left(1-\exp\left(e^{-2cq_y}-e^{-2cq_x}\right)\right).
    \end{align*}
\end{proof}

\subsection{Proof of \cref{thm:cpdag_dag}}

\subsubsection{Auxiliary Results}

\begin{lemma}
    An edge $(x,y)\in\E$ in a DAG $\G=(\V,\E)$ does not form part of an unshielded collider $(x,y,\cdot)$ if and only if $x$ is adjacent to all other parents of $y$.
    \label{lem:undir_adjacency}
\end{lemma}
\begin{proof}
    We show the contrapositive. 
    Assume there exists $z\in\pa(y)$ with $z\neq x$, $(x,z)\notin\E$, and $(z,x)\notin\E$.
    Then, by definition, $(x,y)$ forms part of the unshielded collider $(x, y, z)$. Conversely, if $(x,y)$ forms part of an unshielded collider $(x,y,z)$, then by definition $z\in\pa(y)$, $z\neq x$, and $z$ is not adjacent to $x$.
\end{proof}

\begin{corollary}
    Let $y\in\V$ in a DAG $\G=(\V,\E)$. Define $\U\subseteq\E$ as the set of edges that do not form part of an unshielded collider.
    It holds that $\C_y\coloneq\{x\colon (x,y)\in\U\}$ forms a clique in the skeleton of $\G$.
    \label{cor:undir_clique}
\end{corollary}
\begin{proof}
    By \cref{lem:undir_adjacency}, every $x\in\C_y$ is adjacent to all parents of $y$. The result follows since $\C_y\subseteq\pa(y)$.
\end{proof}

\begin{lemma}[]
    In every clique $\C\subseteq\V$ of a DAG $\G=(\V,\E)$, there exists $c\in\C$ such that $\anc(c)=\anc(\C)$.
    \label{lem:clique_ancestors}
\end{lemma}
\begin{proof}
    Note that the subgraph $\G_\C\subseteq\G$ with nodes $\C$ and all edges between members of $\C$ is a fully connected DAG.
    By definition, in the fully connected DAG $\G_\C$ there exists a sink, that is, $c\in\C$ such that $\C\setminus\{c\}\subseteq\pa(c)$.
    Hence, $c$ inherits the ancestors of all nodes in $\C$ and the statement follows.
\end{proof}

\begin{proposition}
    Let $\G=(\V,\E)$ be a DAG with $\s_\rho(\G)=1$. 
    Every node $y\in\V$ that is not a root node is the center of at least one unshielded collider.
    \label{prop:collider_center}
\end{proposition}
\begin{proof}
    We show the statement via the contrapositive.
    Recall that $\s_\rho(\G)=1$ if and only if $|\rel(a)|<|\rel(b)|$ for all $(a,b)\in\E$.
    Assume there exists a non-root $y\in\V$ that is not the center of any unshielded collider. Since $y$ is a non-root it has at least one parent, and since it is not the center of any unshielded collider, its parents form a clique in the skeleton of $\G$ by \cref{cor:undir_clique}. By \cref{lem:clique_ancestors}, there exists $x\in\pa(y)$ with $\anc(x)=\anc(\pa(y))$.
    By the definition of relatives, it follows that $\rel(x)=\rel(y)$, hence $\s_\rho(\G)\neq1$ and the statement follows. 
\end{proof}

\begin{proposition}
    Let $\G=(\V,\E)$ be a DAG with $\s_\rho(\G)=1$. Let $y\in\V$ be any non-root node.
    Define $\U\subseteq\E$ as the set of edges that do not form part of an unshielded collider.
    Define further $\C_y\coloneq\{x\colon (x,y)\in\U\}$, the set of all parents of $y$ that do not form endpoints of unshielded colliders centered in $y$. Denote as $\D_y\coloneq\pa(y)\setminus\C_y$ the set of all parents that do form endpoints of unshielded colliders centered in $y$.
    For every $c\in\C_y$, there exists $d\in\D_y$ such that $(c,d)\in\E$, meaning $d\in\ch(c)$.
    \label{prop:rule_3_works}
\end{proposition}
\begin{proof}
    We show the statement via the contrapositive.
    Recall that $\s_\rho(\G)=1$ if and only if $|\rel(a)|<|\rel(b)|$ for all $(a,b)\in\E$.
    Assume there exists $c\in\C_y$ that has no $d\in \D_y$ as a child.
    Since any $c\in\C_y$ is adjacent to all other parents of $y$ by \cref{lem:undir_adjacency}, this means that $\D_y\subseteq\pa(c)$.
    By \cref{cor:undir_clique}, the nodes in $\C_y$ form a clique in the skeleton of $\G$. Thus, by \cref{lem:clique_ancestors} there exists $x\in\C_y$ with $\anc(x)=\anc(\C_y)$. Since $\D_y\subseteq\pa(c)$, it follows that $\D_y\cup\C_y=\pa(y)\subseteq\anc(x)$. By the definition of relatives it follows that $\rel(x)=\rel(y)$, hence $\s_\rho(\G)\neq 1$ and the statement follows. 
\end{proof}

\subsubsection{Proof of the theorem}\label{app:proof_thm3}

\begin{proof}
    To prove the statement, we show that all edges in the skeleton of a DAG $\G$ with $\s_\rho(\G)=1$ can be oriented by the orientation rules in \cite{meek1995orientation}.
    Recall that three nodes $a, b, c$ form an unshielded collider $(a,b,c)$ with center $b$ if $(a,b)\in\E$, $(c,b)\in\E$, and $a$ is not adjacent to $c$ \parencite[][Section 2.1.1 S1 and S2]{meek1995orientation}.
    All edges of $\G$ that form part of an unshielded collider can be oriented.
    We define $\U\subseteq\E$ to be the set of edges that cannot be oriented through unshielded colliders.
    We proceed by induction along a topological order of $\G$. Let $y\in\V$ be a non-root node, and assume that all edges into nodes preceding $y$ in this order have been oriented.
    We define the set $\C_y\coloneq\{x\colon (x,y)\in\U\}$ of parents of $y$ that do not form endpoints of a collider centered in $y$, and show that all edges from nodes in $\C_y$ to $y$ can also be oriented.
    We further define the set $\D_y\coloneq\pa(y)\setminus\C_y$.
    By \cref{prop:collider_center}, $y$ is the center of at least one unshielded collider.
    By \cref{prop:rule_3_works}, for any $c\in\C_y$ there exists $d\in\D_y$ such that $(c,d)\in\E$, meaning $c$ is a parent of $d$. 
    Since $d\in\pa(y)$ precedes $y$ in the topological order, the edge $(c,d)$ has been oriented by the induction hypothesis, and $(d,y)$ has been oriented since it forms part of an unshielded collider centered in $y$. Thus, $c\to d\to y$ together with the adjacency of $c$ and $y$ allows orienting $(c,y)$ by Rule 2 in \textcite[][Section 2.1.2]{meek1995orientation}.
    By induction over all non-root nodes, all edges in $\U$ can be oriented and the statement follows. 
\end{proof}

\subsection{Proof of \cref{thm:ts_dags}}\label{app:proof_thm4}

\begin{proof}
    We assume for simplicity that the skeleton of $\G_S$ has a single connected component; otherwise the proof follows straightforwardly by applying the same logic to every connected component. Note that the maximal simple path length in $\G_S$ is at most $p-1$. Since $\G_S$ contains all self-loops, any directed path in $\G_S$ can be padded to arbitrary larger length. Hence, if the left-hand side of \cref{eq:implication} holds, then for all $t>p$, all nodes $i^{(t)}\in\V_T$ have the same relatives by \cref{eq:lemma2_4} in \cref{lem:ancestors}. Thus, every edge $(i^{(t)},j^{(t+1)})\in\E_T$ with $t>p$ contributes $\s_\rho\big((i^{(t)},j^{(t+1)})\big)=1/2$. The number of edges in $\G_T$ with $t\leq p$ for the tail node is at most $p|\E_S|$, whereas $|\E_T|=(T-1)|\E_S|$. Hence the fraction of edges with $t\leq p$ tends to zero as $T\to\infty$, and the statement follows.
\end{proof}

\section{Common Random DAG Models used in the Causal Discovery Literature}\label{app:graphs_lit}
\cref{tab:literature} illustrates the kind of settings used in the literature, 
While not necessarily representative of the whole literature or containing all additional/ablation experiments, the list conveys the scale of synthetic data evaluations likely to be affected by high rel-sortability.
\begin{table}[H]
    \centering
    \small
    \setlength{\tabcolsep}{3pt}  %
    \begin{tabular}{lp{3cm}ll}
        \toprule
        \textbf{Work} & $n$ & $c$ for \textbf{ER}$\left(n, \frac{2c}{n-1}\right)$ & $c$ for \textbf{SF}$(n,c)$ \\[.2em]
        \midrule
        \textcite[][Simulation Tests]{spirtes1995learning} & $\{10\}$ &  $\{1,1.5,2\}$ & - \\
        \textcite[][Section 4.1]{kalisch2007estimating} & $\{7,15,40,70,100\}$ & $\{1,2.5\}$ & - \\
        \textcite[][Section 4]{shimizu2011directlingam} & $\{10,20,50,100\}$ & $\{1,2.5\}$ & - \\
        \textcite[][Section 5.1]{peters2014causal} & $\{4,15\}$ & $\{1\}$ & - \\
        \textcite[][Section 6]{buhlmann2014cam} & $\{10,100\}$ & $\{1,4\}$ & - \\
        \textcite[][Section 7]{hyttinen2014constraint} & $\{6\}$ & $\{1\}$ & - \\
        \textcite[][Section 3.1]{peters2015structural} & $\{5, 20\}$ & $\{0.75\}$ & - \\
        \textcite[][Section 5]{magliacane2016ancestral} & $\{6,7,8\}$ & $\{1\}$ & - \\
        \textcite[][Section 6.1]{huang2018generalized} & $\{10\}$ & $\{1,1.5,2,2.5,3\}$ & -\\
        \textcite[][Section D.1]{zheng2018dags} & $\{20,1000\}$ & $\{1,2,4\}$ & $\{4\}$\\
        \textcite[][Section 4.1]{yu2019dag} & $\{10,20,50,100\}$ & $\{1.5\}$ & -\\
        \textcite[][Section A.6]{lachapelle2019gradient} & $\{10,20,50,100\}$ & $\{1,4\}$ & $\{1,4\}$ \\
        \textcite[][Section 5]{ng2020role} & $\{10,20,50,100\}$ & $\{1,2,4\}$ & $\{1,2,4\}$\\
        \textcite[][Section 4]{brouillard2020differentiable} & $\{10,20,100\}$ & $\{1,4\}$ & -\\
        \textcite[][Sections 4,5]{viinikka2020towards} & $\{20,50\}$ & $\{2\}$ & - \\
        \textcite[][Section 6.1]{lorch2021dibs} & $\{20,50\}$ & $\{2\}$ & $\{2\}$ \\
        \textcite[][Section 6.2]{deleu2022bayesian} & $\{20,50\}$ & $\{2\}$ & - \\
        \textcite[][Section 5]{bello2022dagma} & $\{20,100,200,1000\}$ & $\{4\}$ & $\{4\}$ \\
        \textcite[][Section 4.1, A]{rolland2022score} & $\{10,20,50\}$ & $\{1,4\}$ & $\{1,4\}$ \\
        \textcite[][Sections 6.1]{annadani2023bayesdag} & $\{5,30,50\}$ & $\{1\}$ & $\{1\}$ \\
        \textcite[][Sections 4, C]{montagna2023scalable} & $\{10,20,50,100,\newline200,500,1000\}$ & $\{1,4\}$ & $\{1,4\}$\\
        \textcite[][Sections 3,4,5]{ng2024structure} & $\{15,50\}$ & $\{1,2,4\}$ & - \\
        \textcite[][Sections 5, E]{ormaniec2024standardizing} & $\{20, 60, 100, 140,\newline 180, 220\}$ & $\{2,4\}$ & $\{2,4\}$ \\
        \textcite[][Section 4.2, B]{dhir2025meta} & $\{20\}$ & $\{1,2,3\}$ & $\{1,2,3\}$ \\
        \textcite[][Section 5]{kim2025targeted} & $\{1000\}$ & $\{1,2,3\}$ & $\{1,2,3\}$ \\
        \textcite[][Section 5]{wienobst2025embracing} & $\{50\}$ & $\{4\}$ & $\{4\}$\\
        \bottomrule
    \end{tabular}
    \caption{A selection of random graph settings used for synthetic data experiments in the literature.}
    \label{tab:literature}
\end{table}

\section{Additional Experiments and Details}\label{app:add_exps}

\subsection{Software and licenses}
In all implementation and experiments, we use version $3.10.20$ of the Python programming language. For the experiments, we also rely on the packages listed in \cref{tab:software}.
\begin{table}[H]
\centering
    \begin{tabular}{llll}
    \toprule
    Name & Version & Paper & License \\
    \midrule
    \href{https://numpy.org/}{numpy}
        & 1.26.4 & \cite{harris2020array} & modified BSD \\
    \href{https://scipy.org/}{SciPy}
        & 1.13.1 & \cite{virtanen2020scipy} & BSD 3-clause \\
    \href{https://scikit-learn.org/}{scikit-learn}
        & 1.7.2 & \cite{pedregosa2011scikit} & BSD 3-clause \\
    \href{https://pandas.pydata.org/}{pandas}
        & 2.3.3 & \cite{reback2020pandas} & BSD 3-clause \\
    \href{https://matplotlib.org/}{matplotlib}
        & 3.10.9 & \cite{Hunter:2007} & modified PSF \\
    \href{https://seaborn.pydata.org/}{seaborn}
        & 0.13.2 & \cite{Waskom2021} & BSD 3-clause \\
    \href{https://networkx.org/}{networkx}
        & 3.4.2 & \cite{hagberg2008exploring} & BSD 3-clause \\
    \href{https://igraph.org/}{igraph}
        & 1.0.0 & \cite{csardi2006igraph} & GNU General Public License \\
    \href{https://pypi.org/project/daosim/}{daosim}
        & 0.0.5 & \cite{andrews2024better} & MIT \\
    \href{https://pypi.org/project/UUMCdata/}{UUMC}
        & 0.0 & \cite{herman2025unitless} & GNU General Public License \\
    \href{https://pypi.org/project/flopsearch/}{flopsearch}
        & 0.3.0 & \cite{wienobst2025embracing} & Mozilla Public License 2.0 \\
    \href{https://dagma.readthedocs.io/en/latest/}{DAGMA}
        & 1.1.1 & \cite{bello2022dagma} & Apache \\
    \href{https://pypi.org/project/gadjid/}{gadjid}
        & 0.1.0 & \cite{henckel2024adjustment} & Mozilla Public License 2.0 \\
    \href{https://github.com/uhlerlab/causaldag}{causaldag}
        & 0.1a163 & \cite{squires2018causaldag} & BSD 3-clause \\
    \href{https://github.com/mwien/CliquePicking}{cliquepicking}
        & 0.3.2 & \cite{wienobst2023polynomial} & MIT \\
    \href{https://erdogant.github.io/bnlearn}{bnlearn}
        & 0.14.0 & \cite{taskesen2020bnlearn} & MIT \\
    \href{https://github.com/boschresearch/causalAssembly}{causalAssembly}
        & 1.2.1 & \cite{gobler2024texttt} & AGPL-3.0 \\
    \href{https://github.com/CausalDisco/CausalDisco}{CausalDisco}
        & 0.3.0 & \cite{reisach2023scale} & BSD 3-clause \\
    \bottomrule
    \end{tabular}
    \caption{Software packages used in our experiments.}
    \label{tab:software}
\end{table}

\subsection{MEC sizes}\label{app:mec_sizes}
The excellent results obtained by the FLOP CPDAG discovery algorithm discussed in \cref{sec:emp_causal_discovery} prompt us to investigate the size of the MECs for our simulation settings.
We draw $100$ random DAGs with $n=30$ for each value of the graph density parameter $c\in\{1,\dots,10\}$, and compute the size of the MEC as well as the number and percentage of undirected edges in the CPDAG of each DAG.
\cref{tab:mec_size_exp} shows the mean values for ER DAGs, and \cref{tab:mec_size_exp_SF} for SF DAGs.
For ER DAGs, MECs tend to be small and the fraction of undirected edges tends to be low for $c>1$.
For SF DAGs, mean MEC sizes are roughly twice as large as for ER DAGs for $2\leq c\leq 4$, and smaller otherwise. The percentage of undirected CPDAG edges is similarly low as in ER DAGs and also decreases with $c$.
\begin{table}[H]
    \small
    \centering
    \input{figures/mec_summary_ER_SCM.tex}
    \caption{MEC sizes and undirected CPDAG edges of random ER DAGs.}
    \label{tab:mec_size_exp}
\end{table}
\begin{table}[H]
    \small
    \centering
    \input{figures/mec_summary_SF_SCM.tex}
    \caption{MEC sizes and undirected CPDAG edges of random SF DAGs.}
    \label{tab:mec_size_exp_SF}
\end{table}
If the true MEC is small, then learning the CPDAG well already leaves only a small number of orientations unresolved. These MEC-size results suggest that this is often the case in our simulations, which helps explain why FLOP performs well in DAG-learning metrics.
Hence, in settings like ours, a strong CPDAG learner may be a useful first stage for DAG discovery, after which the remaining orientations could be resolved using additional assumptions, scores, interventions, domain knowledge, or even random guessing for a simple baseline.

\subsection{SF DAG results}\label{app:SF}

This section contains experiment results for SF DAGs corresponding to the ER DAG versions shown in \cref{sec:discovery} of the main text.
Note that there are some additional choices in SF sampling depending on the implementation used. We choose a sampling scheme where in-degree and out-degree increase attachment probability. 
Some other works, for example \cite{zheng2018dags}, choose an SF sampling scheme where only in-degree increases attachment probability, likely further increasing the prevalence of high rel-sortability and small MECs.

\paragraph{Sortabilities across data-generating regimes.}
\cref{fig:schemes_SF} shows the SF DAG experiment corresponding to \cref{fig:schemes_ER} in \cref{sec:schemes} in the main text.
Compared to ER DAGs, we observe a similar trend in rel-sortability with an even closer alignment between oracle and empirical values than for ER DAGs, in particular for DaO.
Though var-sortability is around $1/2$ as for ER DAGs, in SF DAGs we observe very high values of $R^2$-sortability in sparser settings with a decline in denser ones, akin to rel-sortability. This is a reversal of the $R^2$-sortability trend found in ER DAGs, and notably includes DaO and to a lesser extent UUMC, both of which have higher $R^2$-sortabilities for SF DAGs than for ER DAGs.
Overall, we find that all four data generation regimes are characterized by high rel-sortability and $R^2$-sortability, in particular for the density range $1\leq c\leq 4$. 
\begin{figure}[H]
    \centering
    \begin{subfigure}{.24\linewidth}
        \centering
        \includegraphics[width=\linewidth]{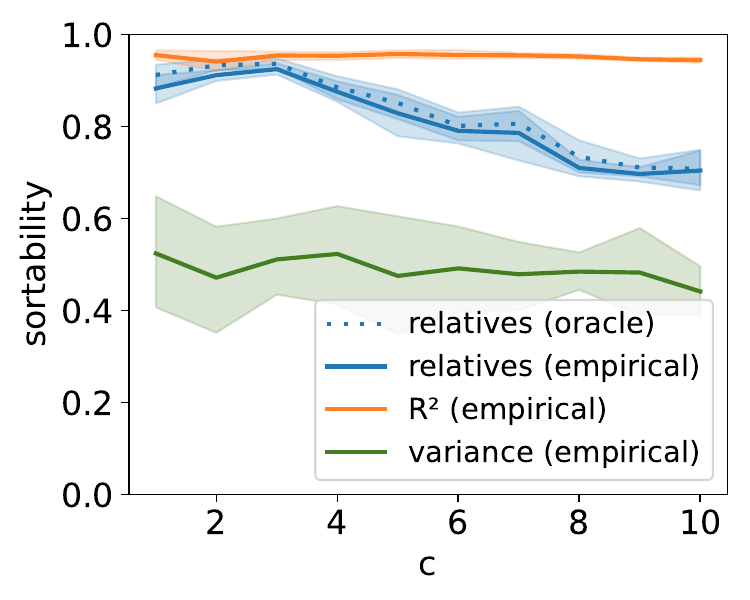}
        \caption*{sSCM}
    \end{subfigure}
    \begin{subfigure}{.24\linewidth}
        \centering
        \includegraphics[width=\linewidth]{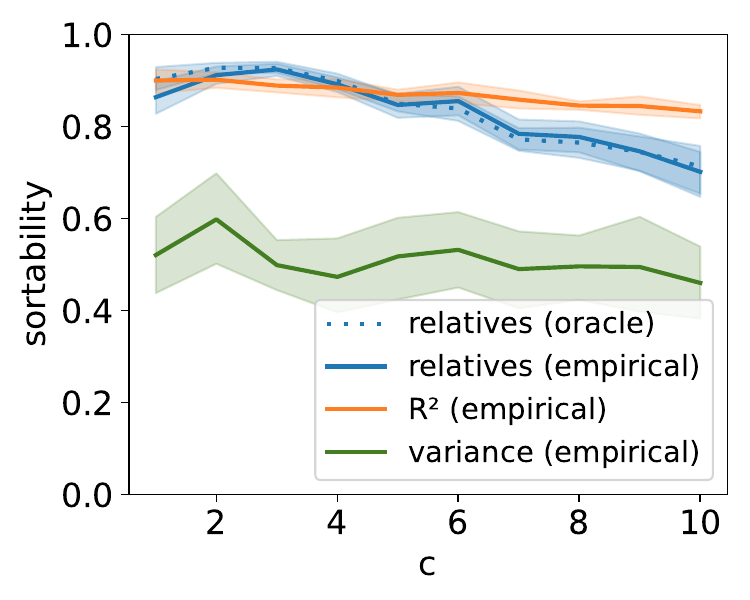}
        \caption*{iSCM}
    \end{subfigure} 
    \begin{subfigure}{.24\linewidth}
        \centering
        \includegraphics[width=\linewidth]{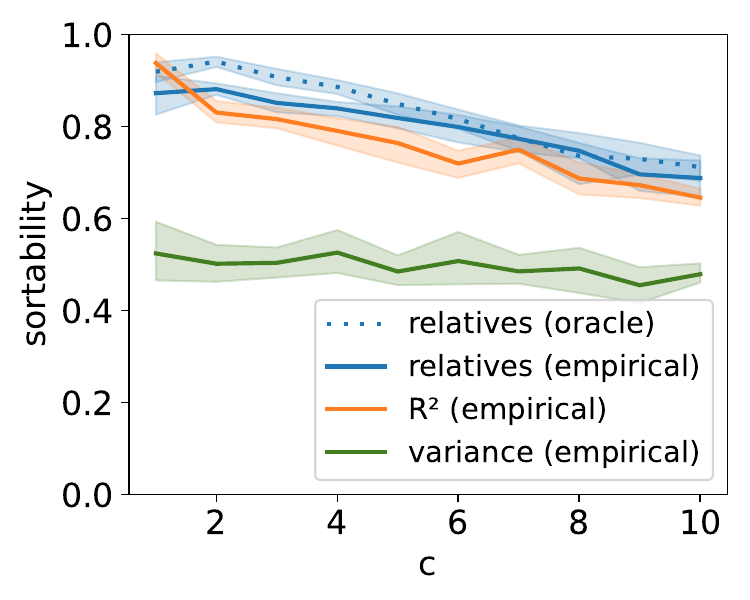}
        \caption*{DaO}
    \end{subfigure}
    \begin{subfigure}{.24\linewidth}
        \centering
        \includegraphics[width=\linewidth]{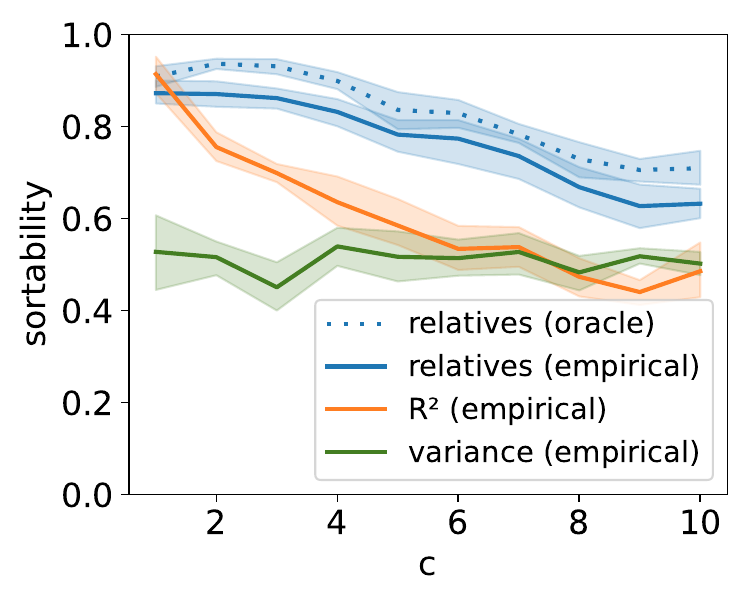}
        \caption*{UUMC}
    \end{subfigure}
    \caption{Comparison of rel-sortability to other sortabilities in SF DAGs.}
    \label{fig:schemes_SF}
\end{figure}

\paragraph{Causal discovery.}
\cref{fig:pAID_SF,fig:SHD_SF} below show the SF DAG experiments complementing the ER DAG experiments shown in \cref{fig:ER_pAID,fig:ER_SHD} of the main text. For parent-AID and compared to ER DAGs, we observe that rel-SortnRegress performs slightly worse than FLOP and $R^2$-SortnRegress on sSCM, but remains a close second to FLOP in all remaining settings.
In SHD, rel-SortnRegress performance also drops on sSCMs, but remains second best on the remaining three sampling procedures for $c\leq4$, the most common density choice in the literature.
\begin{figure}[H]
    \centering
    \begin{subfigure}{.24\linewidth}
        \centering
        \includegraphics[width=\linewidth]{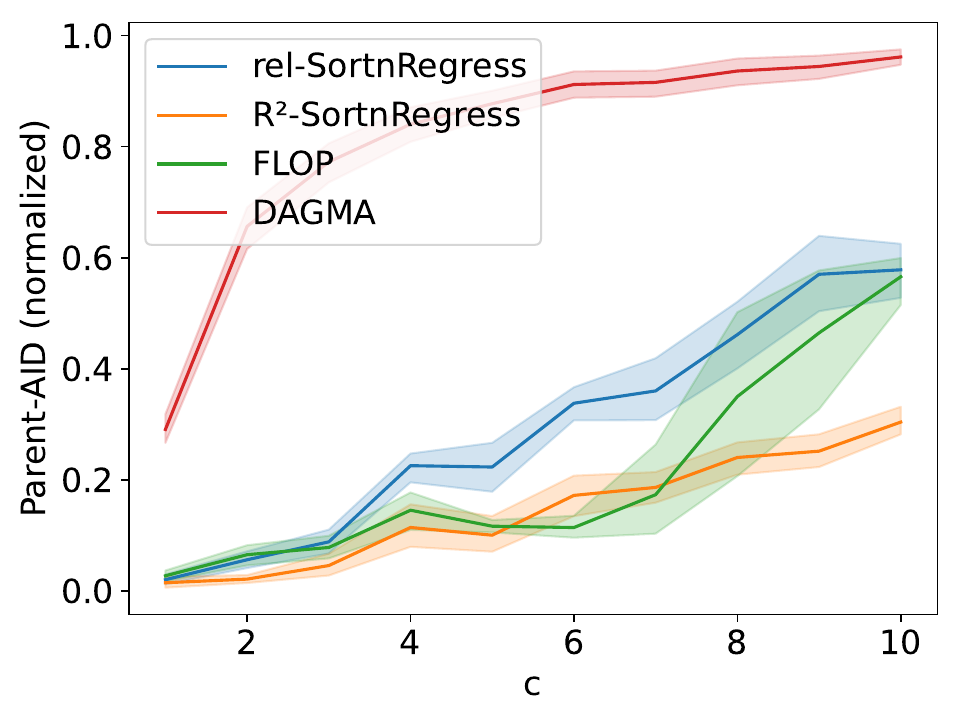}
        \vskip-.5em
        \caption*{sSCM}
    \end{subfigure}
    \begin{subfigure}{.24\linewidth}
        \centering
        \includegraphics[width=\linewidth]{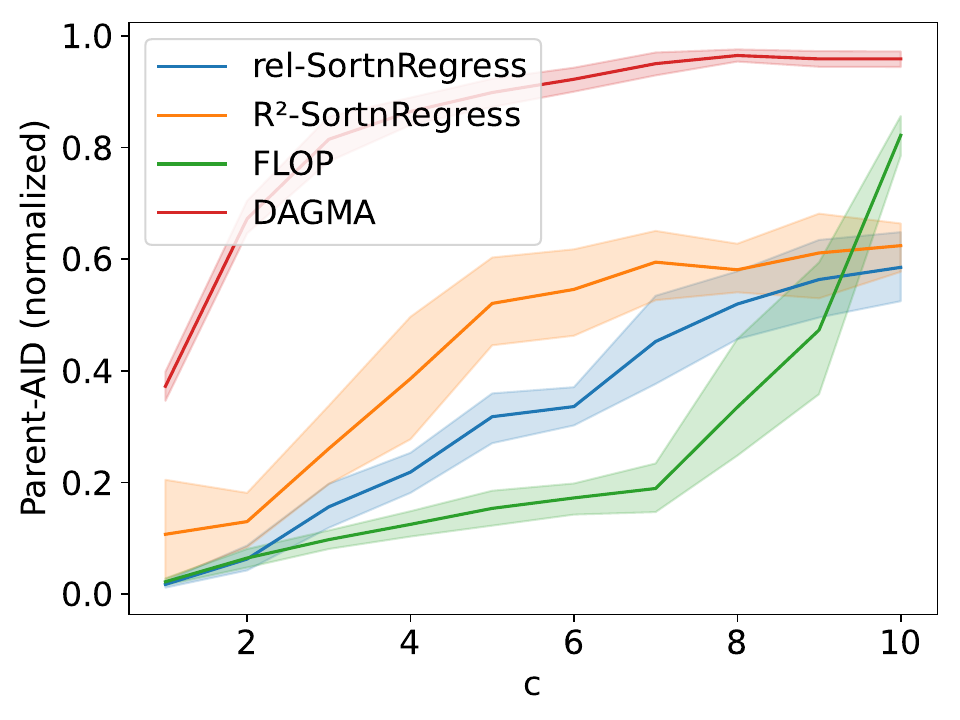}
        \vskip-.5em
        \caption*{iSCM}
    \end{subfigure}
    \begin{subfigure}{.24\linewidth}
        \centering
        \includegraphics[width=\linewidth]{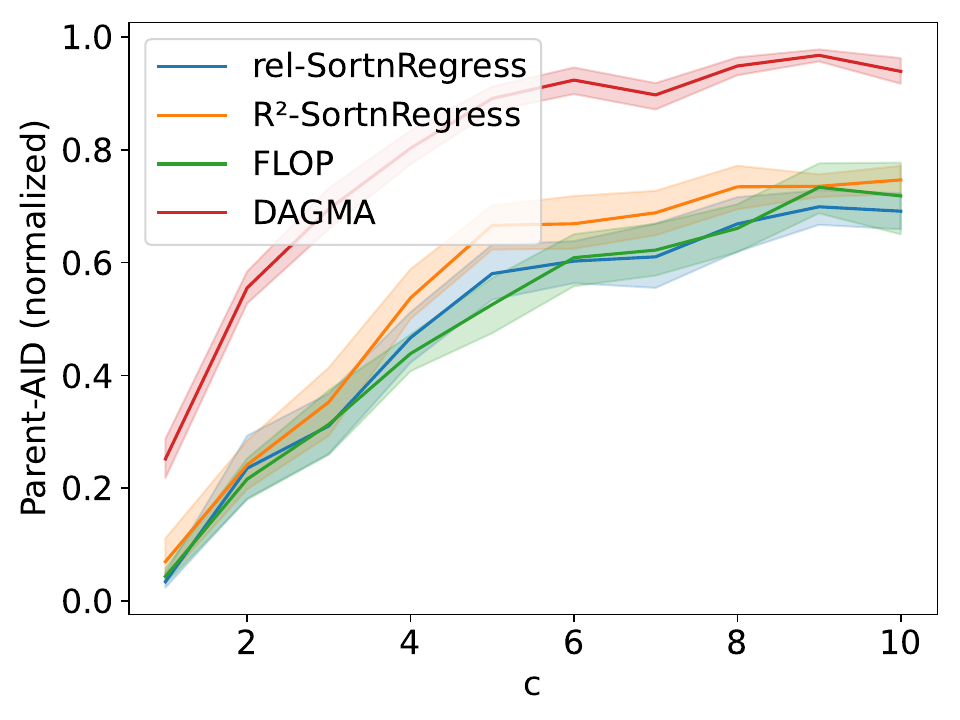}
        \vskip-.5em
        \caption*{DaO}
    \end{subfigure}
    \begin{subfigure}{.24\linewidth}
        \centering
        \includegraphics[width=\linewidth]{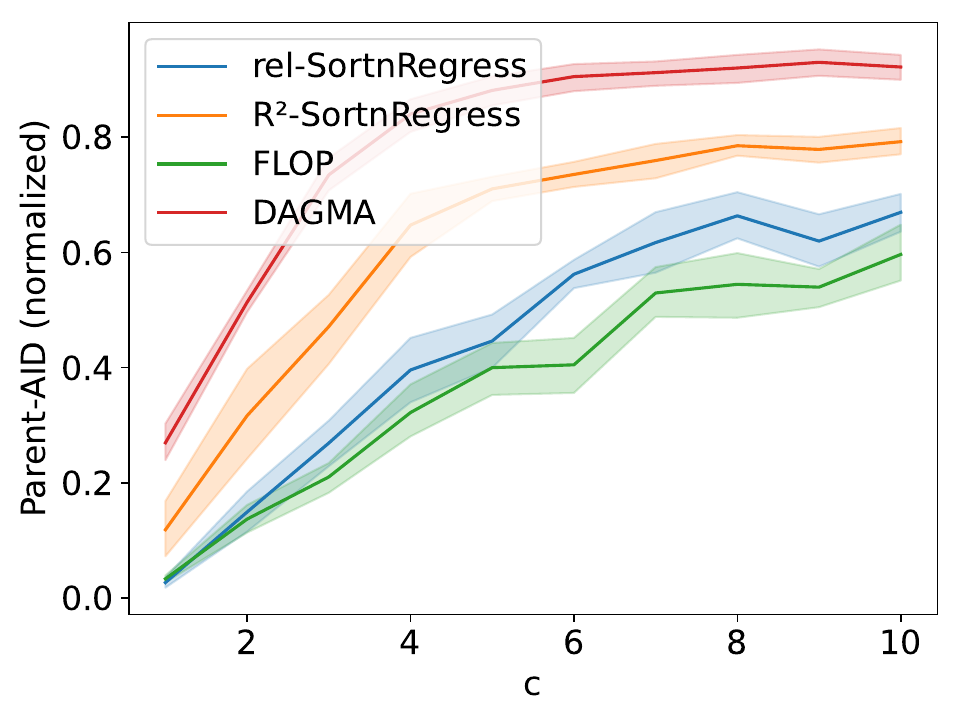}
        \vskip-.5em
        \caption*{UUMC}
    \end{subfigure}
    \caption{Causal discovery results for SF DAGs (parent-AID to true DAG, lower is better).}
    \label{fig:pAID_SF}
\end{figure}
\begin{figure}[H]
    \centering
    \begin{subfigure}{.24\linewidth}
        \centering
        \includegraphics[width=\linewidth]{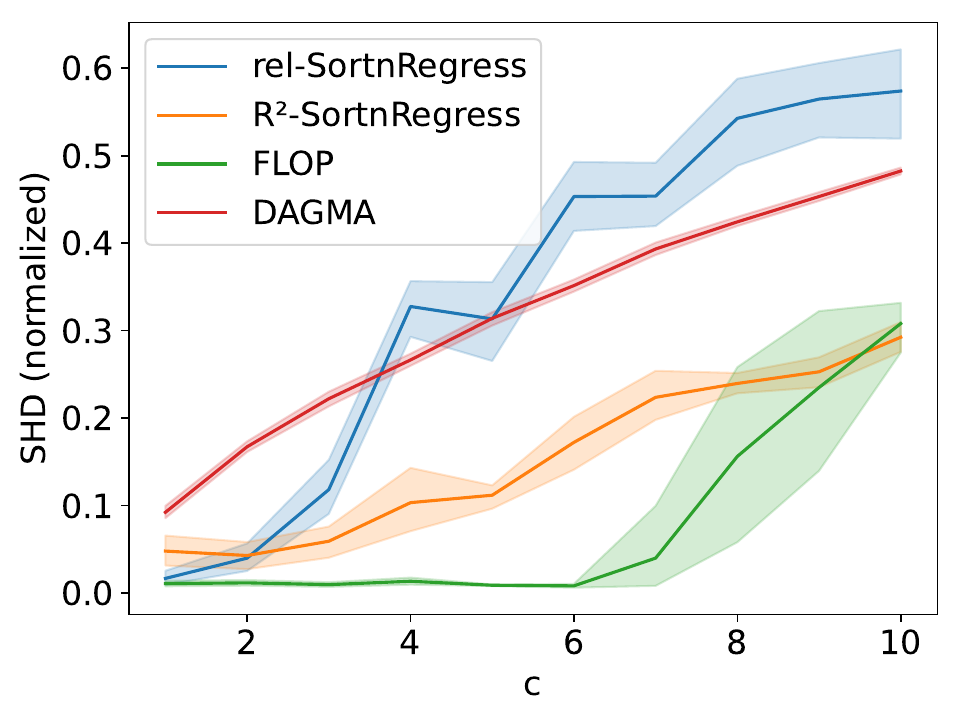}
        \vskip-.5em
        \caption*{sSCM}
    \end{subfigure}
    \begin{subfigure}{.24\linewidth}
        \centering
        \includegraphics[width=\linewidth]{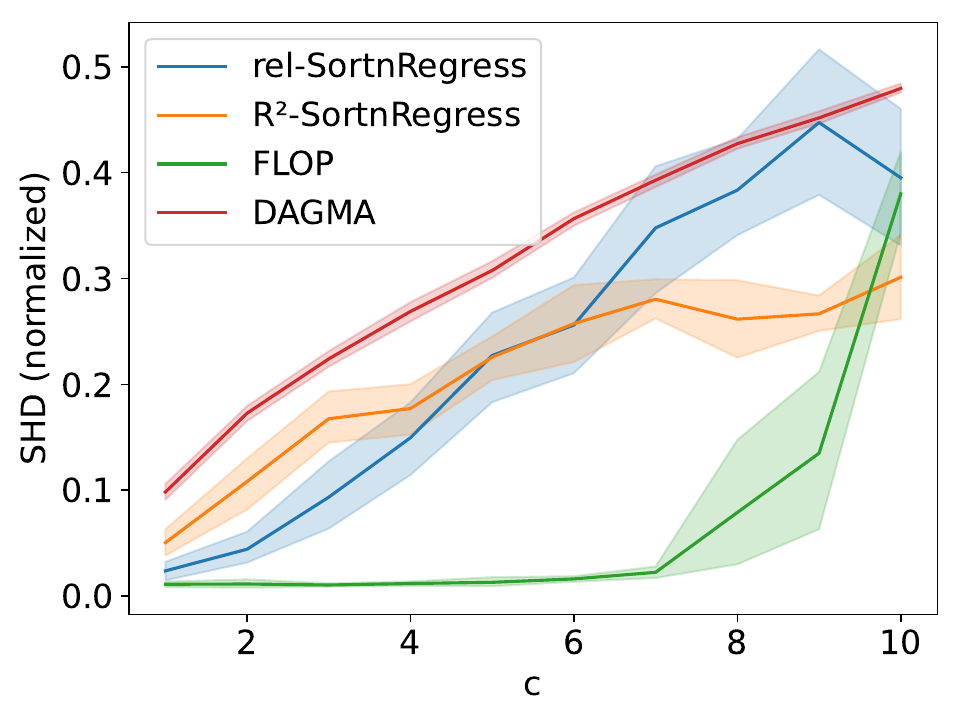}
        \vskip-.5em
        \caption*{iSCM}
    \end{subfigure}
    \begin{subfigure}{.24\linewidth}
        \centering
        \includegraphics[width=\linewidth]{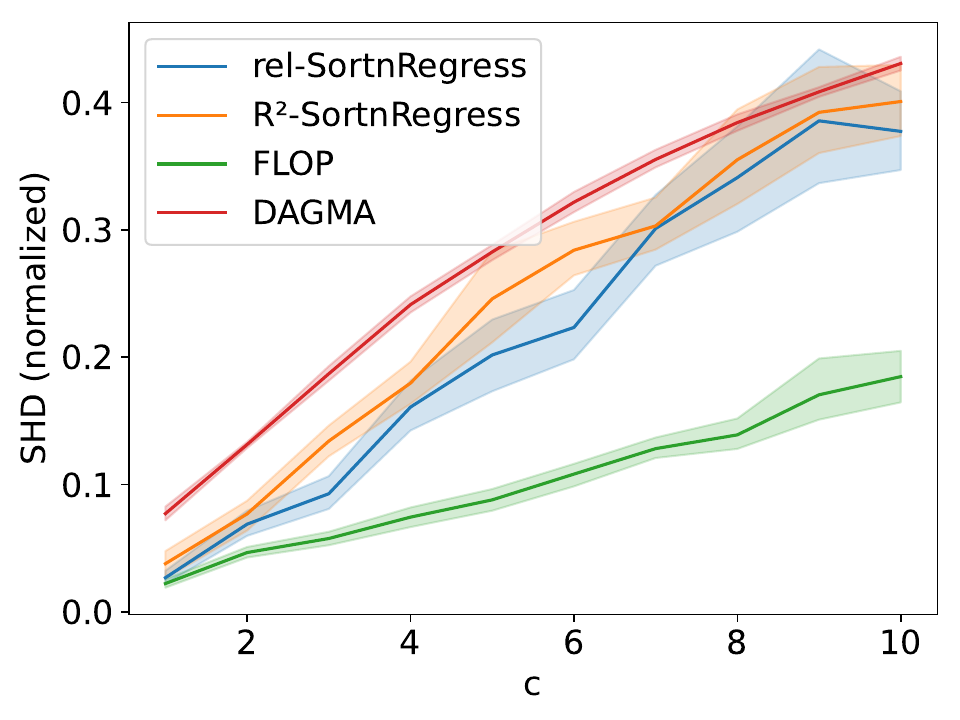}
        \vskip-.5em
        \caption*{DaO}
    \end{subfigure}
    \begin{subfigure}{.24\linewidth}
        \centering
        \includegraphics[width=\linewidth]{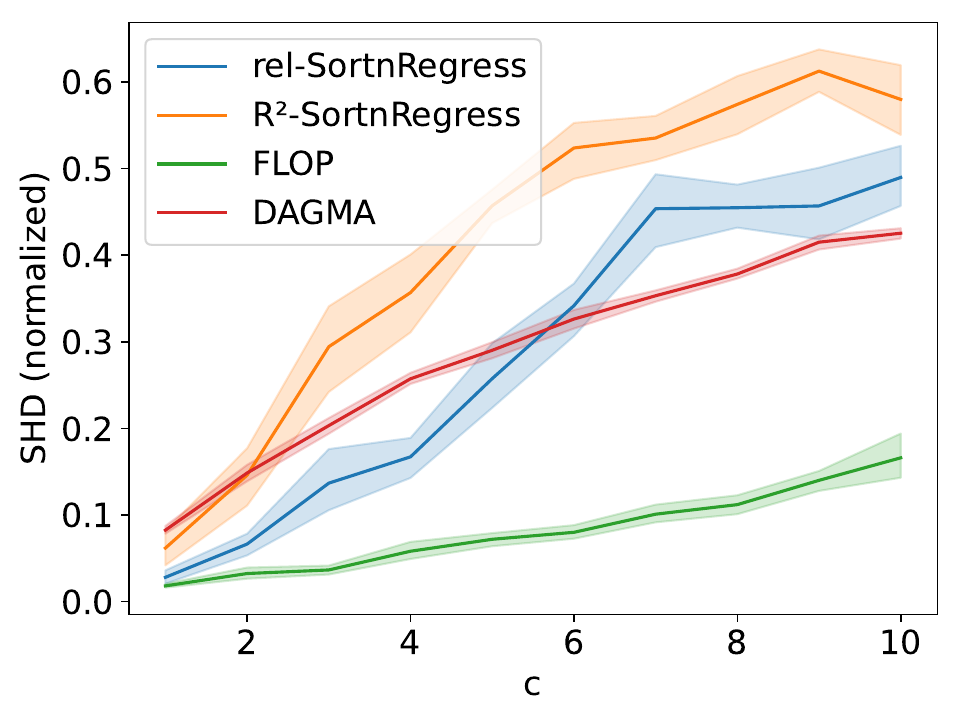}
        \vskip-.5em
        \caption*{UUMC}
    \end{subfigure}
    \caption{Causal discovery results for SF DAGs (SHD to true DAG, lower is better).}
    \label{fig:SHD_SF}
\end{figure}

\subsection{Non-standardized SCMs}\label{app:not-standardized}
\cref{fig:non_std_schemes} compares the different sortabilities on data from non-standardized SCMs as discussed in \cref{sec:schemes}.
We observe that sortability by variance is extremely high for ER and SF DAGs, corroborating the findings in \cite{reisach2021beware}.
\begin{figure}[H]
    \centering
    \begin{subfigure}{.235\linewidth}
        \centering
        \includegraphics[width=\linewidth]{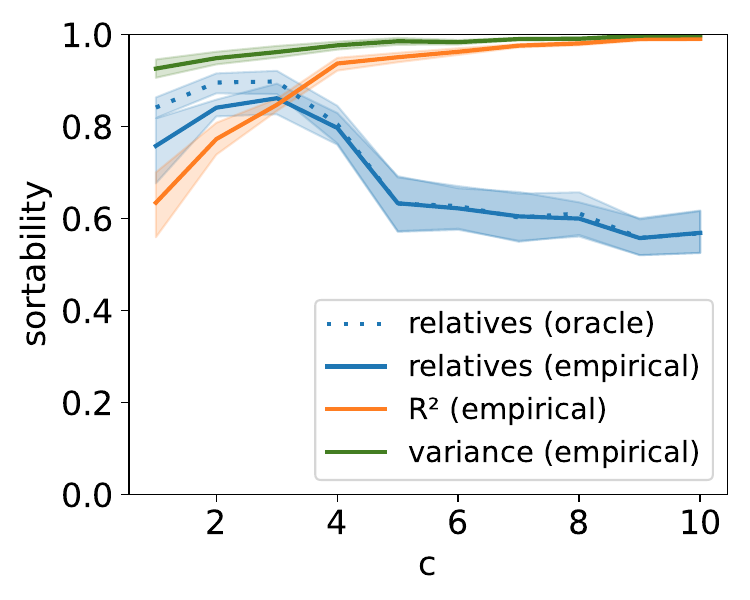}
        \vskip-.5em
        \caption*{ER DAGs}
    \end{subfigure}
    \begin{subfigure}{.235\linewidth}
        \centering
        \includegraphics[width=\linewidth]{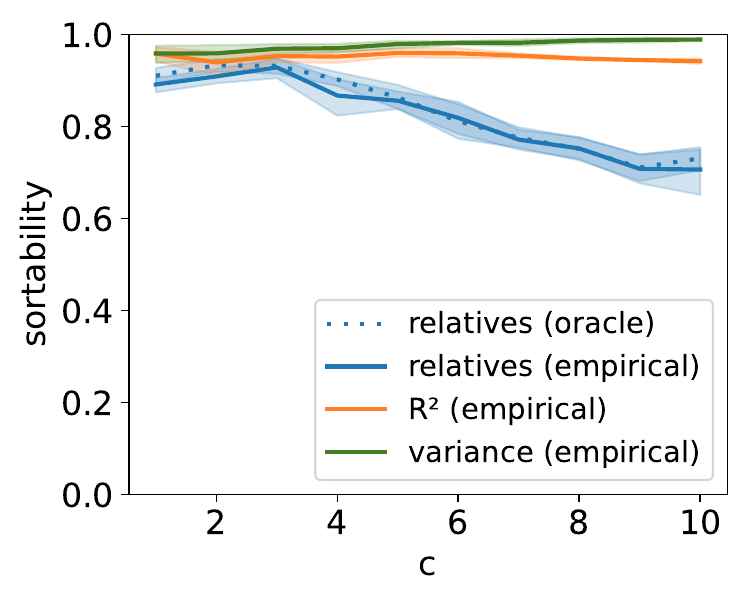}
        \vskip-.5em
        \caption*{SF DAGs}
    \end{subfigure}
    \caption{Sortabilities for non-standardized SCMs.}
    \label{fig:non_std_schemes}
\end{figure}
This suggests information in the variance as an explanation for the performance difference of the DAGMA algorithm between the standardized settings (sSCM panels in \cref{fig:ER_pAID,fig:ER_SHD,fig:pAID_SF,fig:SHD_SF}) and the non-standardized settings in \cref{fig:non_std_discovery}, both in terms of parent-AID (\cref{fig:non_std_pAID}) and SHD (\cref{fig:non_std_SHD}). The other algorithms show behavior very similar to that in the sSCM settings.
\begin{figure}[H]
    \centering
    \begin{subfigure}{.49\linewidth}
        \centering
        \begin{subfigure}{.49\linewidth}
            \centering
            \includegraphics[width=\linewidth]{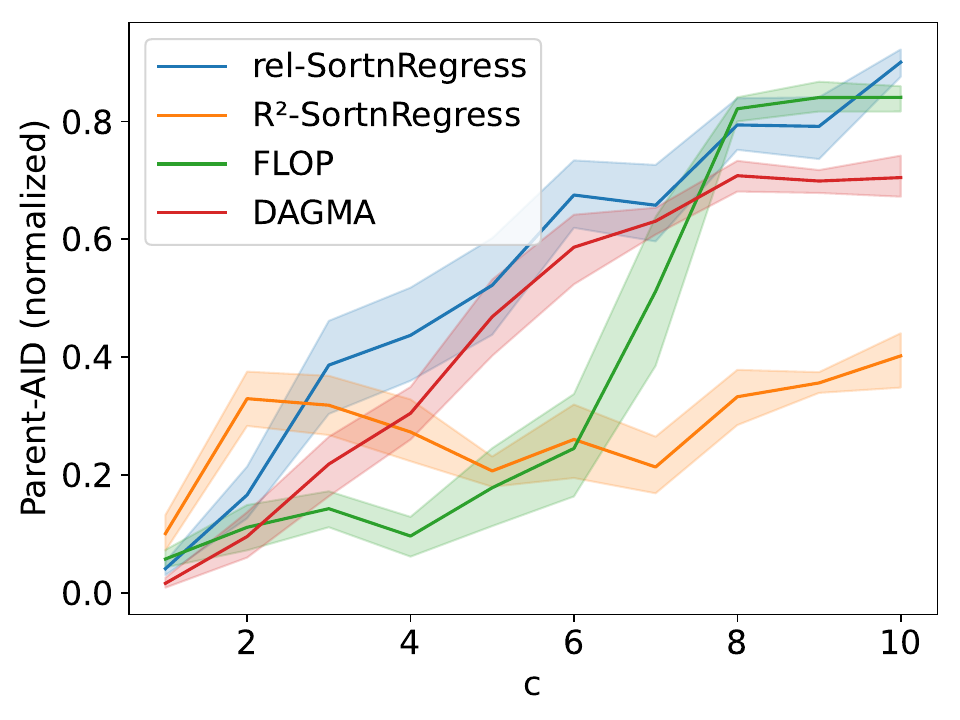}
            \vskip-.5em
            \caption*{ER DAGs}
        \end{subfigure}
        \begin{subfigure}{.49\linewidth}
            \centering
            \includegraphics[width=\linewidth]{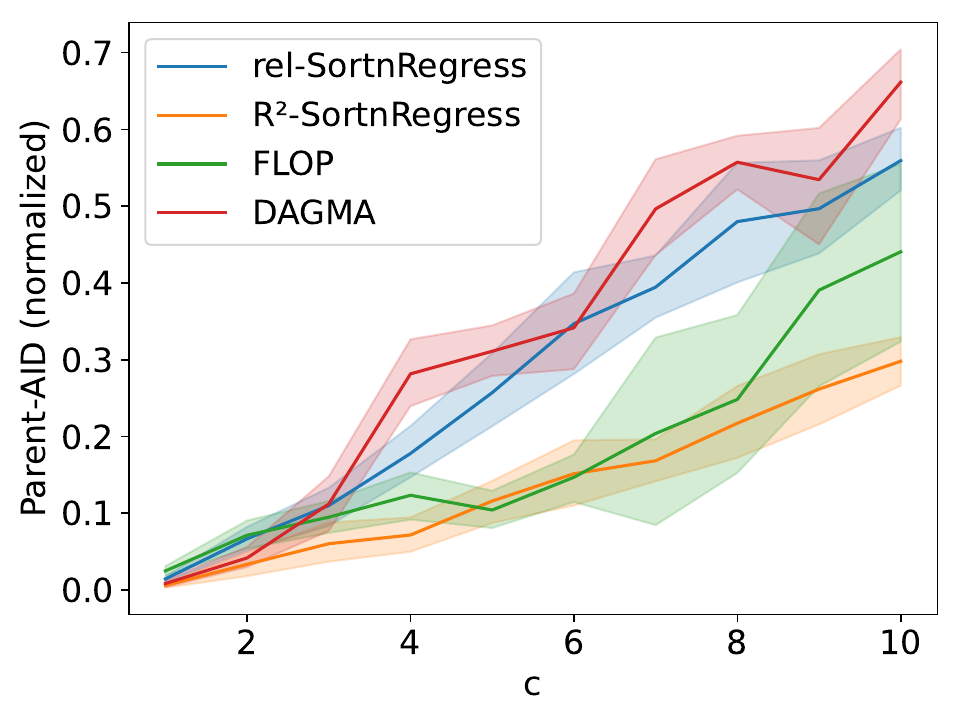}
            \vskip-.5em
            \caption*{SF DAGs}
        \end{subfigure}
        \caption{Parent-AID}
        \label{fig:non_std_pAID}
    \end{subfigure}
    \hfill
    \begin{subfigure}{.49\linewidth}
        \centering
        \begin{subfigure}{.49\linewidth}
            \centering
            \includegraphics[width=\linewidth]{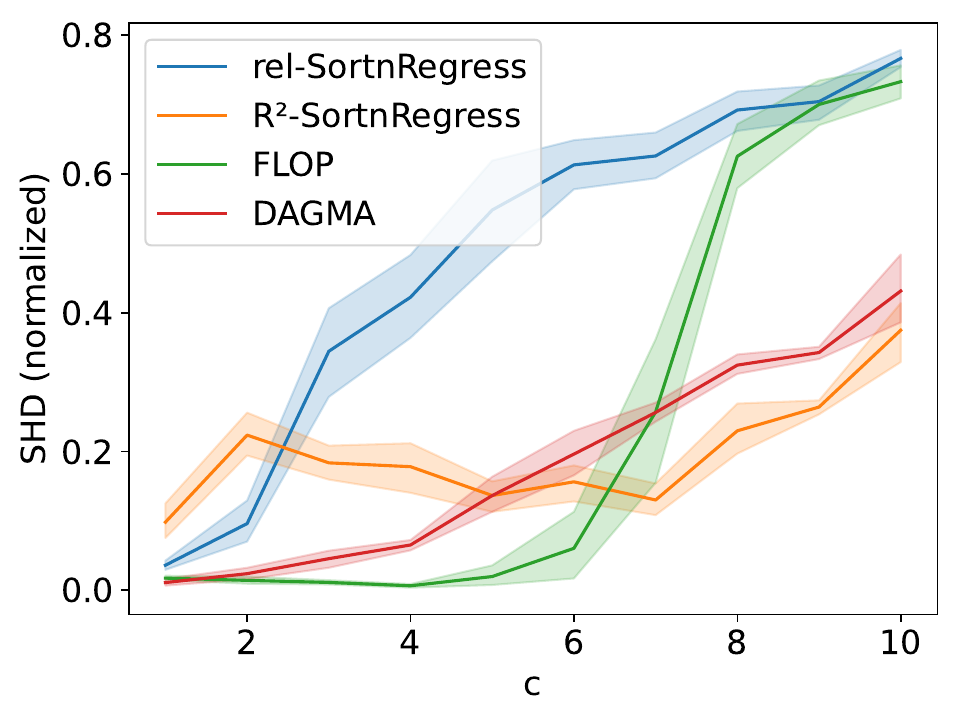}
            \vskip-.5em
            \caption*{ER DAGs}
        \end{subfigure}
        \begin{subfigure}{.49\linewidth}
            \centering
            \includegraphics[width=\linewidth]{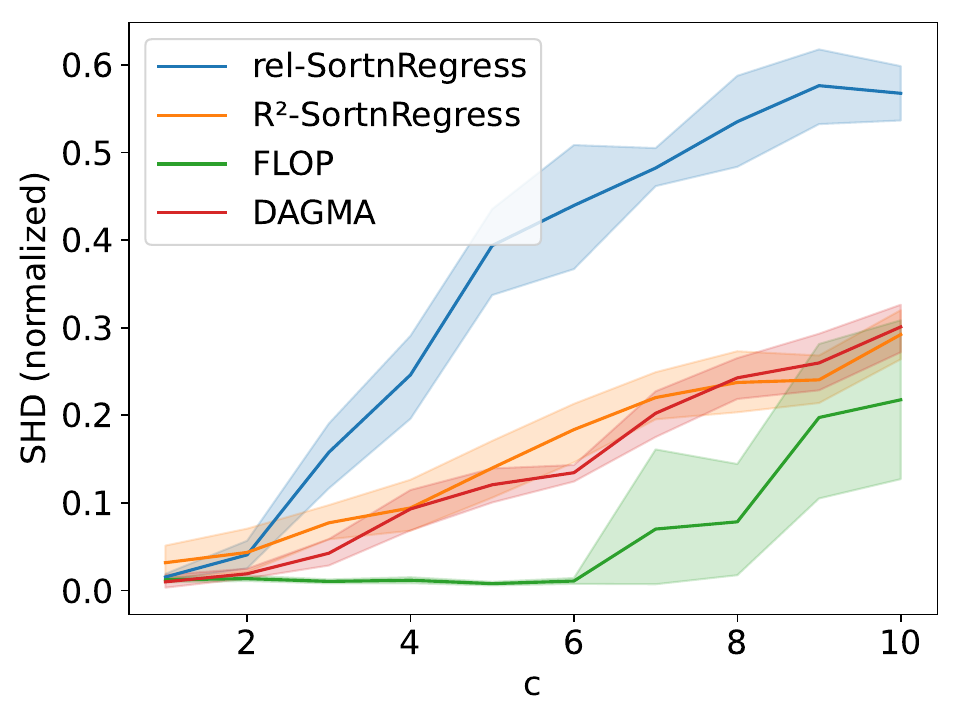}
            \vskip-.5em
            \caption*{SF DAGs}
        \end{subfigure}
        \caption{SHD}
        \label{fig:non_std_SHD}
    \end{subfigure}
    \caption{Causal discovery results (distance to true DAG) for non-standardized SCMs.}
    \label{fig:non_std_discovery}
\end{figure}

\subsection{Comparison to path-based variants of sortability}\label{app:path_length_stb}
\paragraph{Path-length variant.}
The sortabilities of \textcite[][Section 3.1]{reisach2021beware} and \textcite{reisach2023scale} average over paths rather than edges.
For $x,y\in\V$ with $x\neq y$, let $L_{xy}$ denote the number of distinct lengths among the directed paths in $\P_{xy}$. The original definition of sortability, here referred to as the \textit{path-length} variant of sortability, of a graph $\G$ by $\rho$ is given by
\begin{align*}
    \s^{\mathrm{pl}}_\rho(\G)\coloneq\frac{\sum_{x\neq y}L_{xy}\,\s_\rho\bigl((x,y)\bigr)}{\sum_{x\neq y}L_{xy}}.
\end{align*}
Here, $\s_\rho\bigl((x,y)\bigr)$ is defined as in \cref{eq:sortability}, extended to arbitrary node pairs $(x,y)$ not just adjacent ones. Each ancestor--descendant pair thus enters once for every length of directed path connecting it.
\cref{fig:schemes_paths} repeats the experiments of \cref{fig:schemes_ER,fig:schemes_SF} with path-length sortabilities. Compared to the edge-wise sortabilities, rel-sortability is higher for small $c$, yet the qualitative patterns across data generation schemes and edge densities remain the same.
\begin{figure}[H]
    \centering
    \begin{subfigure}{\linewidth}
        \centering
        \begin{subfigure}{.24\linewidth}
            \centering
            \includegraphics[width=\linewidth]{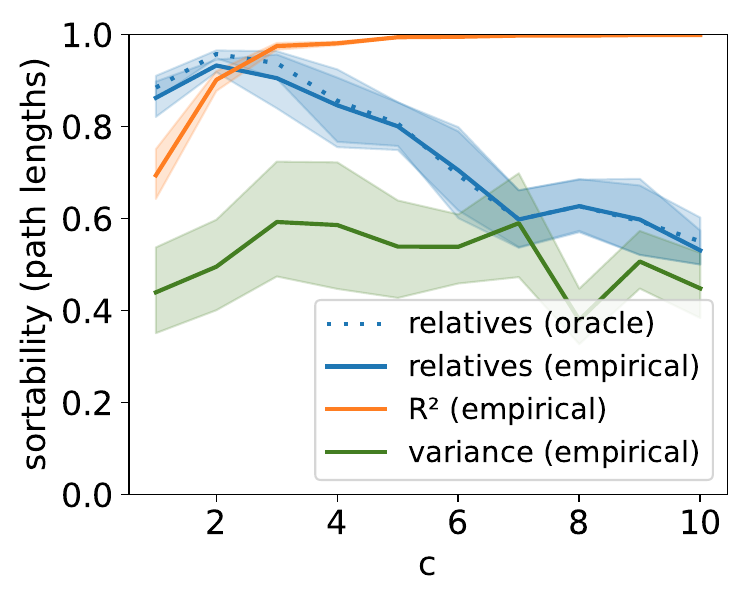}
            \vskip-.5em
            \caption*{sSCM}
        \end{subfigure}
        \begin{subfigure}{.24\linewidth}
            \centering
            \includegraphics[width=\linewidth]{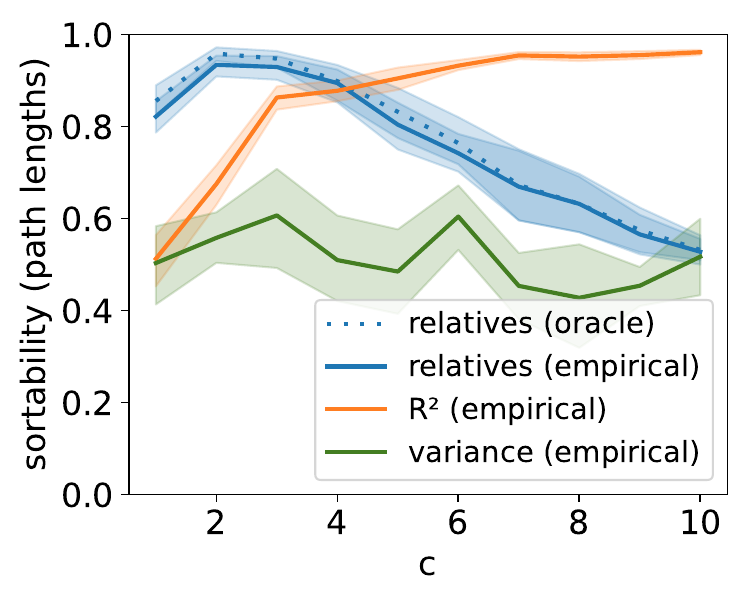}
            \vskip-.5em
            \caption*{iSCM}
        \end{subfigure}
        \begin{subfigure}{.24\linewidth}
            \centering
            \includegraphics[width=\linewidth]{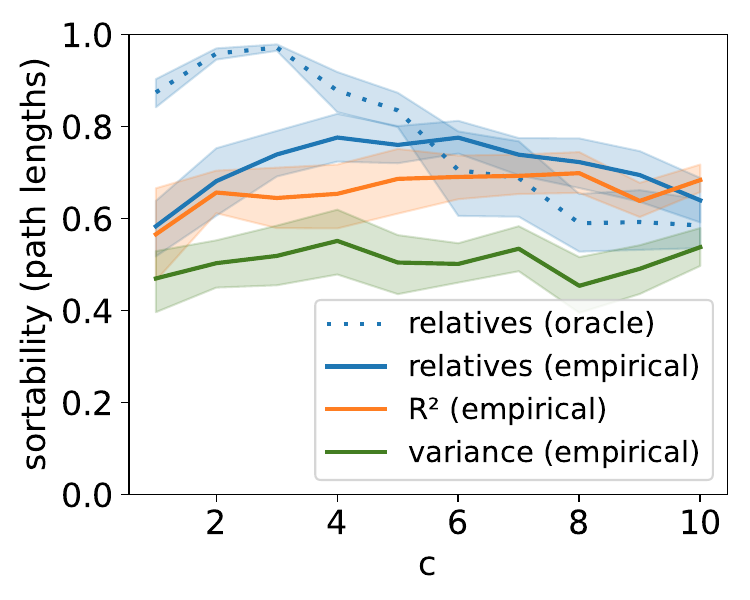}
            \vskip-.5em
            \caption*{DaO}
        \end{subfigure}
        \begin{subfigure}{.24\linewidth}
            \centering
            \includegraphics[width=\linewidth]{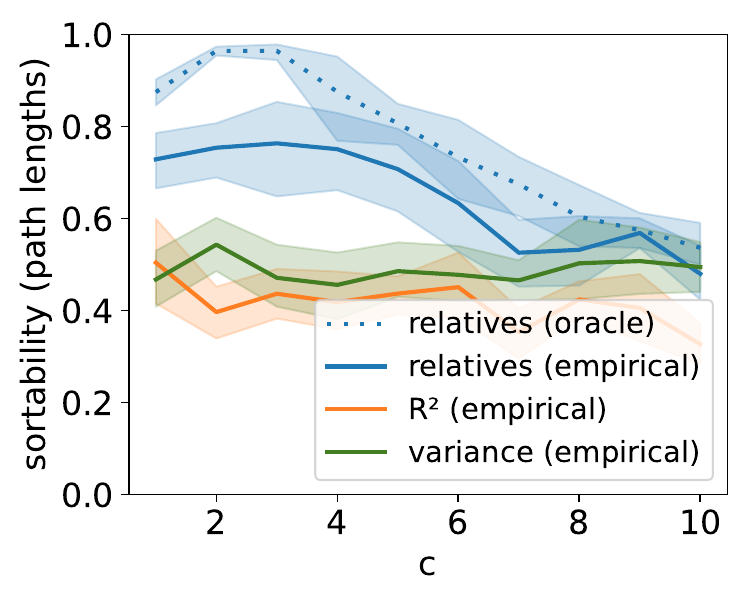}
            \vskip-.5em
            \caption*{UUMC}
        \end{subfigure}
        \caption{ER DAGs}
        \label{fig:schemes_paths_ER}
    \end{subfigure}
    \vskip.5em
    \begin{subfigure}{\linewidth}
        \centering
        \begin{subfigure}{.24\linewidth}
            \centering
            \includegraphics[width=\linewidth]{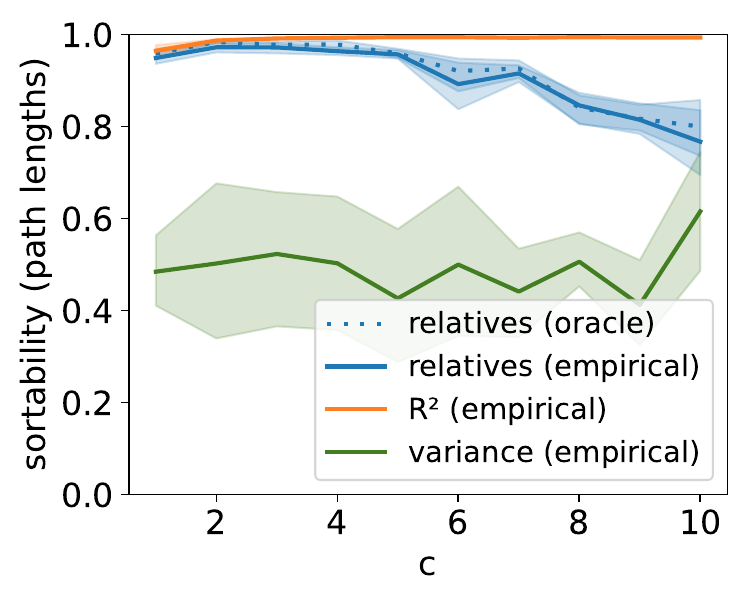}
            \vskip-.5em
            \caption*{sSCM}
        \end{subfigure}
        \begin{subfigure}{.24\linewidth}
            \centering
            \includegraphics[width=\linewidth]{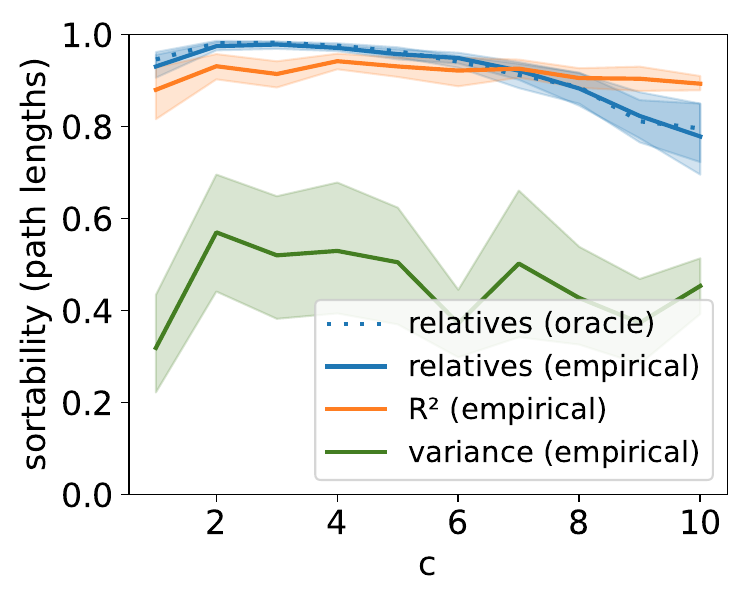}
            \vskip-.5em
            \caption*{iSCM}
        \end{subfigure}
        \begin{subfigure}{.24\linewidth}
            \centering
            \includegraphics[width=\linewidth]{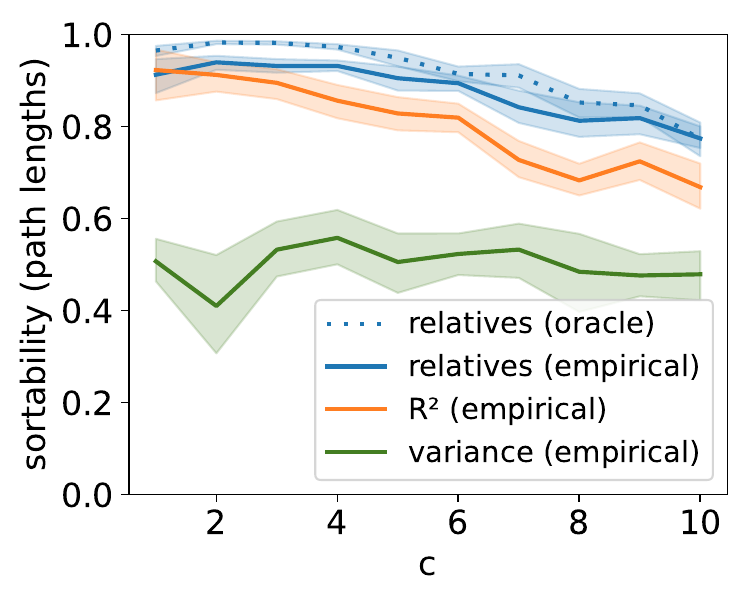}
            \vskip-.5em
            \caption*{DaO}
        \end{subfigure}
        \begin{subfigure}{.24\linewidth}
            \centering
            \includegraphics[width=\linewidth]{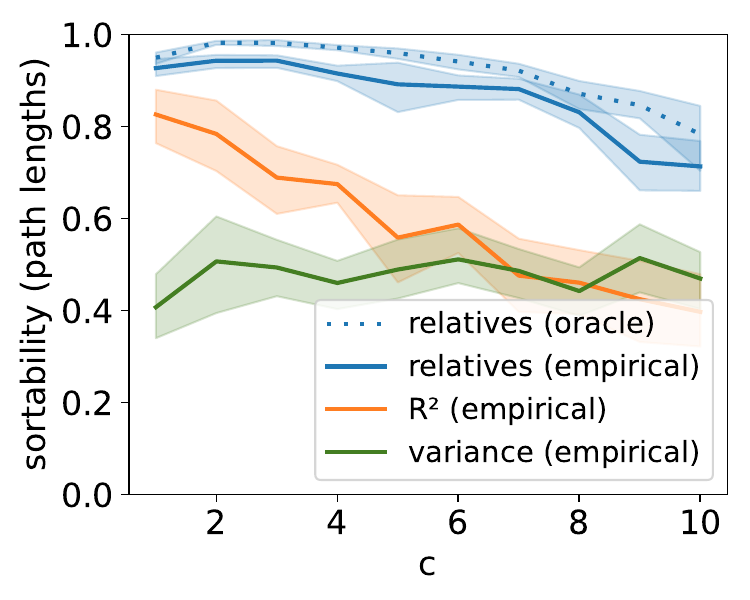}
            \vskip-.5em
            \caption*{UUMC}
        \end{subfigure}
        \caption{SF DAGs}
        \label{fig:schemes_paths_SF}
    \end{subfigure}
    \caption{Path-length variant of the sortabilities in \cref{fig:schemes_ER,fig:schemes_SF}.}
    \label{fig:schemes_paths}
\end{figure}

\paragraph{Path-existence / ancestor variant.}
\textcite[][Section D.1]{reisach2023scale}
also consider a sortability by path-existence, which averages over all ancestor--descendant pairs, counting each pair once. Since the existence of a path marks an ancestor, we refer to it as the \textit{ancestor} variant of sortability, 
\begin{align*}
    \s^{\mathrm{anc}}_\rho(\G)\coloneq\frac{\sum_{y\in\V}\sum_{x\in\anc(y)\setminus\{y\}}\s_\rho\bigl((x,y)\bigr)}{\sum_{y\in\V}|\anc(y)\setminus\{y\}|}.
\end{align*}
This measure quantifies the agreement of an ordering with the causal partial order. Since every ancestor--descendant pair contributes $1$ or $1/2$ by \cref{thm:relatives}, $2\s^{\mathrm{anc}}_\rho(\G)-1$ is the fraction of such pairs that are strictly ordered by their number of relatives.
For completeness, we note that we do not consider the path-count sortability of \textcite[][Section D.2]{reisach2023scale}, since pairs of nodes connected by many paths may dominate this measure.
\cref{fig:schemes_ancestors} shows the experiments of \cref{fig:schemes_ER,fig:schemes_SF} with the ancestor variant of sortability. We observe oracle rel-sortabilities between their edge-wise and path-length variants, and the same qualitative patterns across regimes and settings as before.
\begin{figure}[H]
    \centering
    \begin{subfigure}{\linewidth}
        \centering
        \begin{subfigure}{.24\linewidth}
            \centering
            \includegraphics[width=\linewidth]{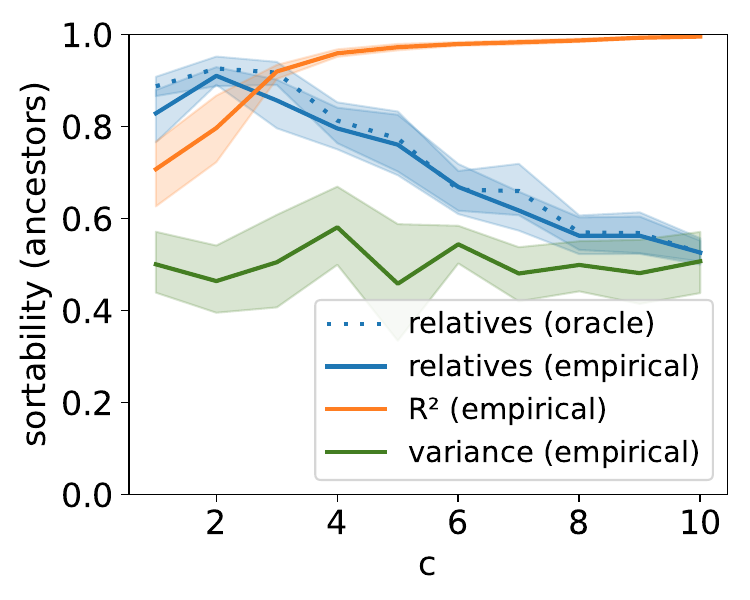}
            \vskip-.5em
            \caption*{sSCM}
        \end{subfigure}
        \begin{subfigure}{.24\linewidth}
            \centering
            \includegraphics[width=\linewidth]{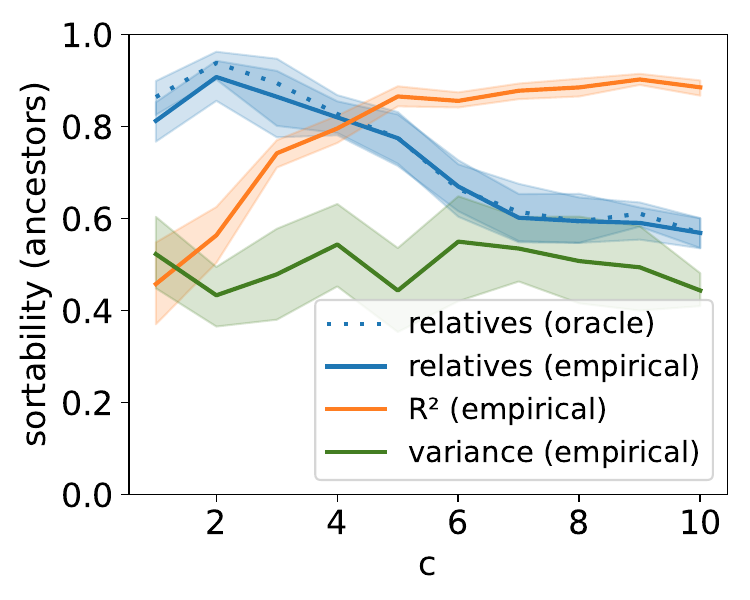}
            \vskip-.5em
            \caption*{iSCM}
        \end{subfigure}
        \begin{subfigure}{.24\linewidth}
            \centering
            \includegraphics[width=\linewidth]{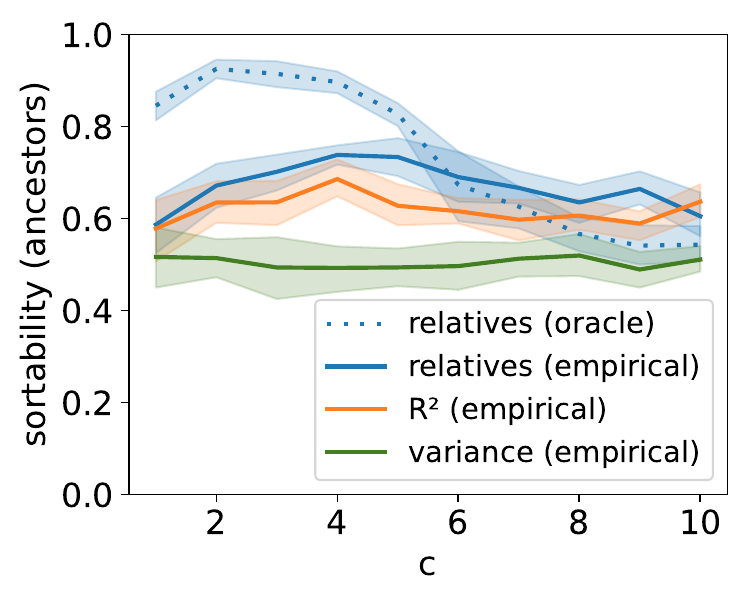}
            \vskip-.5em
            \caption*{DaO}
        \end{subfigure}
        \begin{subfigure}{.24\linewidth}
            \centering
            \includegraphics[width=\linewidth]{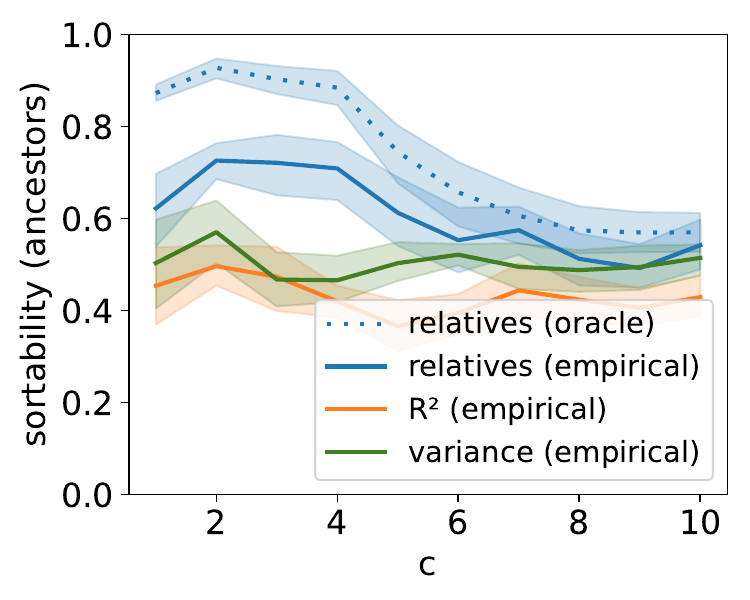}
            \vskip-.5em
            \caption*{UUMC}
        \end{subfigure}
        \caption{ER DAGs}
        \label{fig:schemes_ancestors_ER}
    \end{subfigure}
    \vskip.5em
    \begin{subfigure}{\linewidth}
        \centering
        \begin{subfigure}{.24\linewidth}
            \centering
            \includegraphics[width=\linewidth]{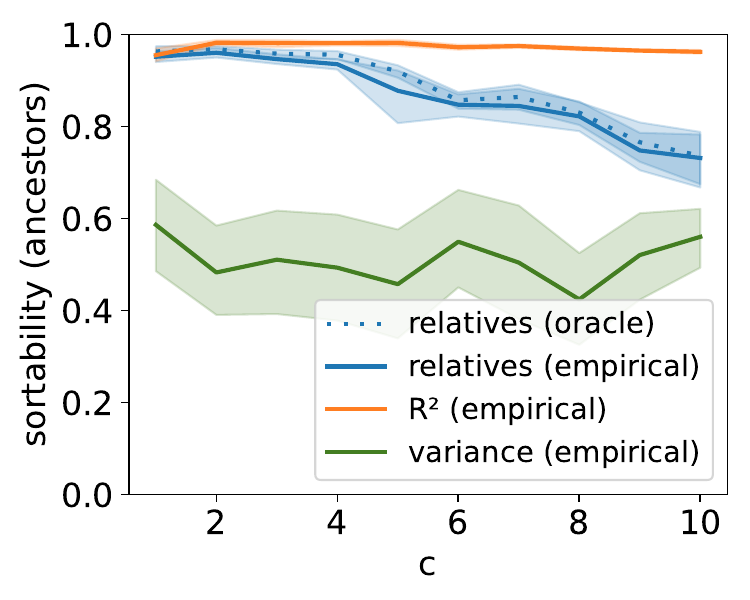}
            \vskip-.5em
            \caption*{sSCM}
        \end{subfigure}
        \begin{subfigure}{.24\linewidth}
            \centering
            \includegraphics[width=\linewidth]{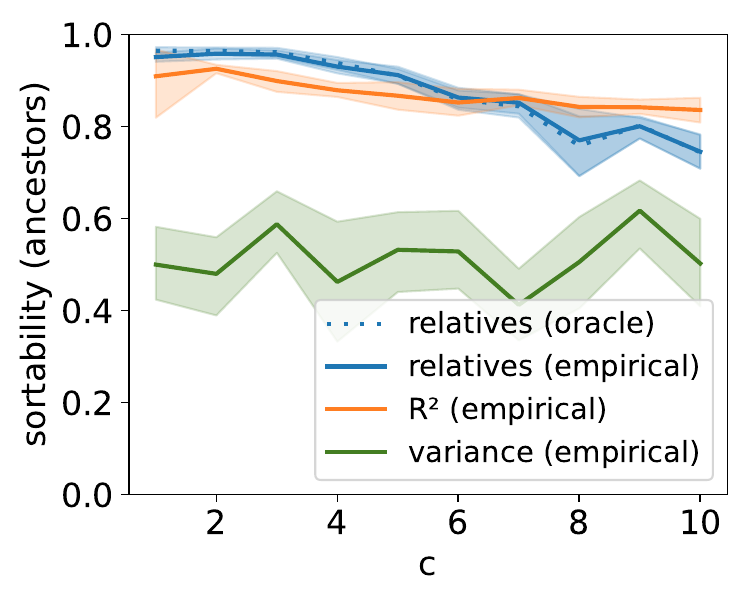}
            \vskip-.5em
            \caption*{iSCM}
        \end{subfigure}
        \begin{subfigure}{.24\linewidth}
            \centering
            \includegraphics[width=\linewidth]{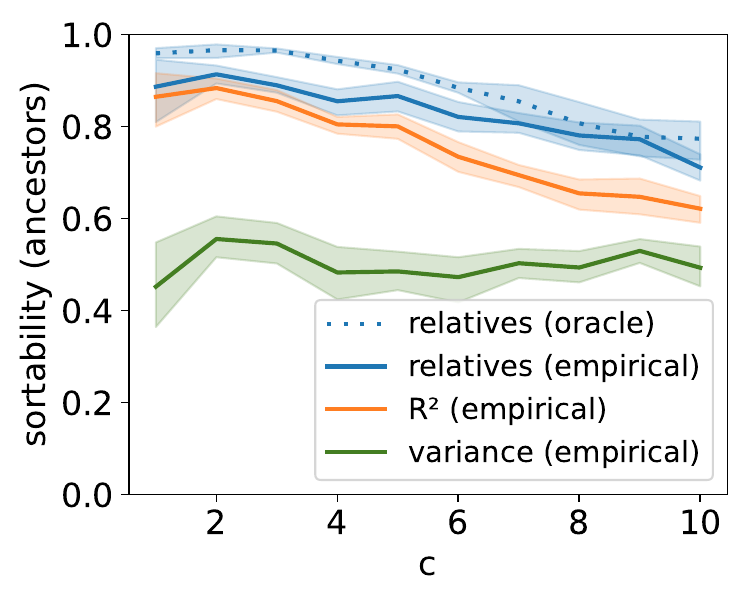}
            \vskip-.5em
            \caption*{DaO}
        \end{subfigure}
        \begin{subfigure}{.24\linewidth}
            \centering
            \includegraphics[width=\linewidth]{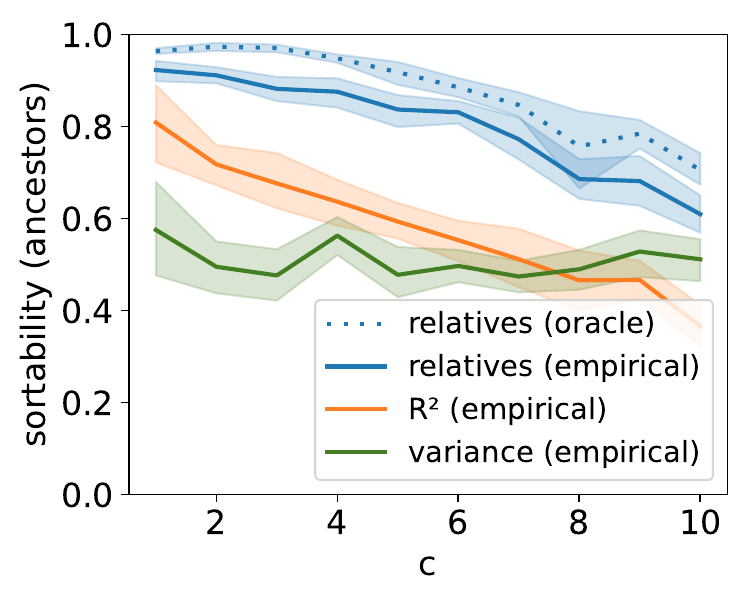}
            \vskip-.5em
            \caption*{UUMC}
        \end{subfigure}
        \caption{SF DAGs}
        \label{fig:schemes_ancestors_SF}
    \end{subfigure}
    \caption{Ancestor variant of the sortabilities in \cref{fig:schemes_ER,fig:schemes_SF}.}
    \label{fig:schemes_ancestors}
\end{figure}

\subsection{Rel-sortability of time-unrolled graphs}\label{app:ts_dags}
\cref{fig:ts_emp} below illustrates \cref{thm:ts_dags} for finite time horizons. We draw $100$ strongly connected summary graphs with $p=10$ nodes through rejection sampling of binary adjacency matrices with off-diagonal elements drawn from $\text{Bernoulli}(0.3)$, and diagonal elements fixed at $1$. We then unroll them for different time horizons $T$, and compute rel-sortability of the resulting DAGs with $n=pT$ nodes. The blue line shows the mean rel-sortability. The dotted horizontal line is placed at the chance-level sortability value of $1/2$. 
Unlike the other plots, the light blue shaded area here shows the $90\%$ percentile interval (the $95\%$ bootstrap confidence interval of the mean shown in the other line plots would be visually indistinguishable from the mean itself).
We observe that rel-sortability approaches $1/2$, as predicted by the theorem.
\begin{figure}[H]
    \centering
    \includegraphics[width=.5\linewidth,trim={0 1.5em 0 0},clip]{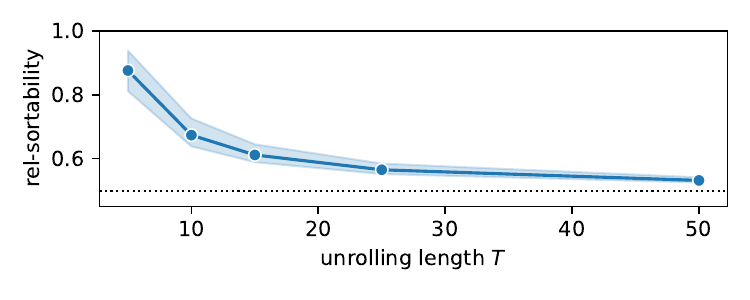}
    \caption{Rel-sortability of time-unrolled DAGs.}
    \label{fig:ts_emp}
\end{figure}
\noindent \cref{tab:ts_summary} lists the corresponding mean, median, and $5\%$ and $95\%$ quantiles of rel-sortability across the $100$ graphs for each $T$. As $T$ grows, the mean decreases towards $1/2$ and the spread between the quantiles narrows.
\begin{table}[H]
    \centering
    \input{figures/ts_summary.tex}
    \caption{Rel-sortability of time-unrolled DAGs for different time horizons $T$.}
    \label{tab:ts_summary}
\end{table}

%% file: figures/mec_summary_ER_SCM.tex
\begin{tabular}{lrrrrrrrrrr}
\toprule
metric & $c=1$ & $c=2$ & $c=3$ & $c=4$ & $c=5$ & $c=6$ & $c=7$ & $c=8$ & $c=9$ & $c=10$ \\
\midrule
MEC size & 154.4 & 28.2 & 16.8 & 11.8 & 7.9 & 11.7 & 13.1 & 10.0 & 10.7 & 24.9 \\
undirected & 8.1 & 5.5 & 4.3 & 3.9 & 3.4 & 3.5 & 3.8 & 3.5 & 3.8 & 5.3 \\
\% undirected & 27.8 & 9.3 & 4.9 & 3.3 & 2.3 & 2.0 & 1.8 & 1.5 & 1.4 & 1.8 \\
\bottomrule
\end{tabular}

%% file: figures/mec_summary_SF_SCM.tex
\begin{tabular}{lrrrrrrrrrr}
\toprule
metric & $c=1$ & $c=2$ & $c=3$ & $c=4$ & $c=5$ & $c=6$ & $c=7$ & $c=8$ & $c=9$ & $c=10$ \\
\midrule
MEC size & 36.3 & 65.0 & 30.4 & 21.8 & 11.5 & 10.4 & 9.9 & 7.7 & 7.9 & 5.3 \\
undirected & 4.5 & 5.9 & 5.0 & 4.3 & 3.7 & 3.5 & 3.5 & 3.1 & 3.3 & 2.6 \\
\% undirected & 15.5 & 10.5 & 6.1 & 4.1 & 3.0 & 2.4 & 2.2 & 1.8 & 1.7 & 1.3 \\
\bottomrule
\end{tabular}

%% file: figures/ts_summary.tex
\begin{tabular}{rrrrrr}
\toprule
$T$ & Nodes $n=pT$ & Mean & 5\% & Median & 95\% \\
\midrule
5 & 50 & 0.876 & 0.811 & 0.873 & 0.938 \\
10 & 100 & 0.674 & 0.638 & 0.668 & 0.725 \\
15 & 150 & 0.612 & 0.589 & 0.608 & 0.645 \\
25 & 250 & 0.565 & 0.552 & 0.563 & 0.585 \\
50 & 500 & 0.532 & 0.525 & 0.531 & 0.541 \\
\bottomrule
\end{tabular}